\documentclass[11pt]{article}

\usepackage{fullpage}
\usepackage{amsthm,amssymb,amsmath}  
\usepackage{xspace,enumerate}
\usepackage[utf8]{inputenc}
\usepackage{thmtools}
\usepackage{thm-restate}
\usepackage{authblk}
\usepackage[hypertexnames=false,colorlinks=true,urlcolor=Blue,citecolor=Green,linkcolor=BrickRed]{hyperref}

  \theoremstyle{plain}
  \newtheorem{theorem}{Theorem}
  \newtheorem{lemma}[theorem]{Lemma}

  \newtheorem{observation}[theorem]{Observation}
  \newtheorem{definition}[theorem]{Definition}

\title{Almost Optimal Distance Oracles for Planar Graphs}

\author[1,2]{Panagiotis Charalampopoulos}
\author[3]{Paweł Gawrychowski}
\author[2]{Shay Mozes}
\author[4]{Oren Weimann}

\affil[1]{
Department of Informatics, King's College London, UK\\
\href{mailto:panagiotis.charalampopoulos@kcl.ac.uk}{panagiotis.charalampopoulos@kcl.ac.uk}}

\affil[2]{
Efi Arazi School of Computer Science, The Interdisciplinary Center Herzliya, Israel\\
\href{mailto:smozes@idc.ac.il}{smozes@idc.ac.il}}

\affil[3]{
Institute of Computer Science, University of Wrocław, Poland\\
\href{mailto:gawry@mimuw.edu.pl}{gawry@mimuw.edu.pl}}

\affil[4]{
Department of Computer Science, University of Haifa, Israel\\
\href{mailto:oren@cs.haifa.ac.il}{oren@cs.haifa.ac.il}
}

\date{\vspace{-5ex}}

\usepackage[dvipsnames,usenames]{color}

\usepackage[ruled]{algorithm}
\usepackage[noend]{algpseudocode}
\algnewcommand{\LineComment}[1]{\State \(\triangleright\) #1}

\usepackage[margin=1in]{geometry}
\usepackage{graphicx}
\usepackage[noadjust]{cite}
  \usepackage{microtype}
    \usepackage{xspace}
    \usepackage{amsmath,amsfonts}
    \usepackage{wasysym}
  \usepackage{makecell}
  \usepackage{tabularx}
  \usepackage{comment}
  \usepackage{todonotes}
  \usepackage{thmtools}
  \usepackage{thm-restate}
  \usepackage[capitalise]{cleveref}
  \usepackage{verbatim}
\usepackage{units}
\usepackage{color}
\definecolor{darkblue}{rgb}{0,0.08,0.45}
  \usepackage{epsfig}
    \usepackage{graphics}
\usepackage{pgf,tikz}
  \usetikzlibrary{trees, arrows, shapes, snakes}
\usepackage{marvosym}
  \usepackage[caption=false]{subfig}
\usetikzlibrary{decorations.pathmorphing}
\usetikzlibrary{decorations.markings}
\usetikzlibrary{decorations.pathmorphing,shapes}
\usetikzlibrary{calc,decorations.pathmorphing,shapes}
\usepackage{subfig}
\usepackage{booktabs}
\usepackage{etoolbox}
\usepackage[title]{appendix}
\usepackage{enumitem}

\newif\iffull
\fulltrue

\newif\ifspreport
\spreporttrue

\DeclareMathOperator*{\argmin}{argmin}

\newcommand{\cO}{\mathcal{O}}
\newcommand{\TG}{\mathcal{T}}
\newcommand{\AG}{\mathcal{A}}
\newcommand{\cOtilde}{\tilde{\cO}}
\newcommand{\Vor}{\textsf{Vor}}
\newcommand{\VD}{\textsf{VD}}
\newcommand{\out}[1]{#1^{out}}
\newcommand{\VDin}{\textsf{VD}_{in}}
\newcommand{\VDout}{\textsf{VD}_{out}}
\newcommand{\weight}{{\rm \omega}}

\begin{document}

\title{Almost Optimal Distance Oracles for Planar Graphs\thanks{This work was partially supported by the Israel Science Foundation under grant number 592/17.}}
\maketitle

\thispagestyle{empty}

\begin{abstract}
We present new tradeoffs between space and query-time for exact distance oracles in directed weighted planar graphs. 
These tradeoffs are almost optimal in the sense that they are within polylogarithmic, sub-polynomial or arbitrarily small polynomial factors from the na\"{\i}ve linear space, constant query-time lower bound.
These tradeoffs include:  
(i) an oracle with space\footnote{The $\cOtilde(\cdot)$ notation hides polylogarithmic factors.} $\cOtilde(n^{1+\epsilon})$ and query-time $\cOtilde(1)$ 
for any constant $\epsilon>0$, 
(ii) an oracle with space $\cOtilde(n)$ 
and query-time $\cOtilde(n^{\epsilon})$ for any constant $\epsilon>0$, and
(iii) an oracle with space $n^{1+o(1)}$ and query-time $n^{o(1)}$. 
\end{abstract}

\definecolor{uququq}{rgb}{0.25,0.25,0.25}
\definecolor{xdxdff}{rgb}{0.49,0.49,1}

  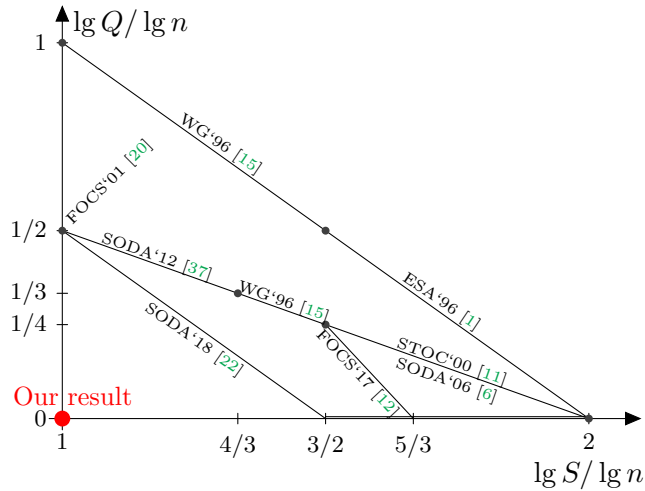
\begin{figure}[h!]
    \begin{center}
\begin{tikzpicture}[line cap=round,line join=round,>=triangle 45,x=7cm,y=5cm]
\draw[->,color=black] (-0.025,0) -- (1.1,0);
\draw[color=black] (0.333333333333333,2pt) -- (0.333333333333333,-2pt) node[below] {\footnotesize $4/3$};
\draw[color=black] (1,2pt) -- (1,-2pt) node[below] {\footnotesize $2$};
\draw[color=black] (0.5,2pt) -- (0.5,-2pt) node[below] {\footnotesize $3/2$};
\draw[color=black] (0.6666666666666,2pt) -- (0.66666666666666,-2pt) node[below] {\footnotesize $5/3$};
\draw[color=black] (0.0,2pt) -- (0.0,-2pt) node[below] {\footnotesize $1$};

\draw[color=black] (2pt,0.25) -- (-2pt,0.25) node[left] {\footnotesize $1/4$};
\draw[color=black] (2pt,0.333333333333333) -- (-2pt,0.333333333333333) node[left] {\footnotesize $1/3$};
\draw[color=black] (2pt,0.5) -- (-2pt,0.5) node[left] {\footnotesize $1/2$};
\draw[color=black] (2pt,1) -- (-2pt,1) node[left] {\footnotesize $1$};
\draw[color=black] (2pt,0) -- (-2pt,0) node[left] {\footnotesize $0$};
\draw[->,color=black] (0,-0.03) -- (0,1.1);
\draw[color=black] (1,-0.15) node {$\lg S/\lg n$};
\draw[color=black] (0.13,1.05) node {$\lg Q/\lg n$};
\clip(-0.1,-0.1) rectangle (1.1,1.1);
\draw (0,1)-- (1,0);
\draw (0,0.5)-- (0.333333333333333,0.3333333333333);
\draw (0.333333333,0.33333333333)-- (0.5,0.25);
\draw (0.5,0.25)-- (1,0);
\draw (0.5,0.25)-- (0.66666666666666,0);
\draw (0.66666666666666,0.005) -- (1,0.005); 
\draw (0,0.5) -- (0.5,0);
\draw (0.5,0.005) -- (0.6666666666,0.005);
\fill [color=uququq] (0,1.0) circle (1.5pt);
\fill [color=uququq] (1.0,0) circle (1.5pt);
\fill [color=uququq] (0,0.5) circle (1.5pt);
\draw[color=black] (0.09,0.63) node[rotate=45] {\tiny FOCS`01 \cite{FR}};
\draw[color=black] (0.30,0.73) node[rotate=-34] {\tiny WG`96 \cite{DBLP:conf/wg/Djidjev96}};
\fill [color=uququq] (0.5,0.5) circle (1.5pt);
\draw[color=black] (0.72,0.31) node[rotate=-34] {\tiny ESA`96 \cite{DBLP:conf/esa/ArikatiCCDSZ96}};
\fill [color=uququq] (0.5,0.25) circle (1.5pt);
\fill [color=uququq] (0.33333333333333,0.333333333333333) circle (1.5pt);
\draw[color=black] (0.42,0.31) node[rotate=-19] {\tiny WG`96 \cite{DBLP:conf/wg/Djidjev96}};
\draw[color=black] (0.74,0.15) node[rotate=-19] {\tiny STOC`00 \cite{DBLP:conf/stoc/ChenX00}};
\draw[color=black] (0.73,0.10) node[rotate=-19] {\tiny SODA`06 \cite{DBLP:journals/algorithmica/Cabello12}};
\draw[color=black] (0.56,0.12) node[rotate=-47] {\tiny FOCS`17 \cite{DBLP:conf/focs/Cohen-AddadDW17}};
\draw[color=black] (0.18,0.43) node[rotate=-19] {\tiny SODA`12 \cite{DBLP:conf/soda/MozesS12}};
\draw[color=black] (0.25,0.21) node[rotate=-35] {\tiny SODA`18 \cite{vorexact}};

\fill [color=red] (0,0) circle (3pt);
\draw[color=red] (0.02,0.06) node {\small Our result};

\end{tikzpicture}
    \end{center}
    \caption{Tradeoff of the space ($S$) vs.~the query time ($Q$) for exact distance oracles in planar graphs on a doubly logarithmic scale, hiding subpolynomial factors. The previous tradeoffs are indicated by solid black lines and points, while our new tradeoff is indicated by the red point at the bottom left.}\label{fig:tradeoffs}
  \end{figure}
  
\clearpage
\setcounter{page}{1}

\section{Introduction}

Computing shortest paths is one of the most fundamental and well-studied algorithmic problems, with numerous applications in various fields.
In the data structure version of the problem, the goal is to preprocess a graph into a compact representation such that the distance (or a shortest path) between any pair of vertices can be retrieved efficiently.
Such data structures are called \emph{distance oracles}.
Distance oracles are useful in applications ranging from navigation, geographic information systems and logistics, to computer games, databases, packet routing, web search, computational biology, and social networks. The topic has been studied extensively; see for example the survey by Sommer~\cite{DBLP:journals/csur/Sommer14} for a comprehensive overview and references.

The two main measures of efficiency of a distance oracle are the space it occupies and the time it requires to answer a distance query. To appreciate the tradeoff between these two quantities consider two na\"{\i}ve oracles. 
The first stores all $\Theta(n^2)$ pairwise distances in a table, and answers each query in constant time using table lookup. The second only stores the input graph, and runs a shortest path algorithm over the entire graph upon each query. 
Both of these oracles are not adequate when working with mildly large graphs. 
The first consumes too much space, and the second is too slow in answering queries. 
A third quantity of interest is the preprocessing time required to construct the oracle. 
Since computing the data structure is done offline, this quantity is often considered less important. 
However, one approach to dealing with dynamic networks is to recompute the entire data structure quite frequently, which is only feasible when the preprocessing time is reasonably small.

One of the ways to design oracles with small space is to consider {\em approximate distance oracles} (allowing a small stretch in the distance output). However, it turns out that one cannot get both small stretch and small space. 
In their seminal paper, Thorup and Zwick~\cite{DBLP:journals/jacm/ThorupZ05} showed that, assuming the girth conjecture of Erd\H{o}s~\cite{Erdos}, there exist dense 
graphs for which no oracle with size less than $n^{1+1/k}$ and stretch less than $2k-1$ exists.
 P\v{a}tra\c{s}cu  and Roditty~\cite{DBLP:journals/siamcomp/PatrascuR14} showed that even sparse graphs with $\cOtilde(n)$ edges do not have distance oracles with stretch better than 2 and subquadratic space, conditioned on a widely believed conjecture on the hardness of set intersection.  To bypass these impossibility results one can impose additional structure on the graph.  
In this work we follow this approach and focus on distance oracles for planar graphs. 

\paragraph{Distance oracles for planar graphs.}

The importance of studying distance oracles for planar graphs stems from several reasons. First, distance oracles for planar graphs are ubiquitous in real-world applications such as geographical navigation on road networks~\cite{DBLP:journals/csur/Sommer14} (road networks are often theoretically modeled as planar graphs even though they are not quite planar due to tunnels and overpasses). Second, shortest paths in planar graphs exhibit a remarkable wealth of structural properties that can be exploited to obtain efficient oracles. Third, techniques developed for shortest paths problems in planar graphs often carry over (because of intricate and elegant connections) to maximum flow problems. Fourth, planar graphs have proved to be an excellent sandbox for the development of algorithms and techniques that extend to broader families of graphs. 

As such, distance oracles for planar graphs have been extensively studied. Works on exact distance oracles for planar graphs started in the 1990's with oracles requiring $\cOtilde(n^2/Q)$ space and $\cOtilde(Q)$ query-time for any $Q \in [1,n]$~\cite{DBLP:conf/wg/Djidjev96,DBLP:conf/esa/ArikatiCCDSZ96}. 
Note that this includes the two trivial approaches mentioned above.  
Over the past three decades, many other works presented exact distance oracles for planar graphs with increasingly better space to query-time tradeoffs~\cite{DBLP:conf/wg/Djidjev96,DBLP:conf/esa/ArikatiCCDSZ96,DBLP:conf/stoc/ChenX00,FR,DBLP:conf/wads/Nussbaum11,DBLP:conf/soda/MozesS12,DBLP:journals/algorithmica/Cabello12,DBLP:conf/focs/Cohen-AddadDW17,vorexact}. Figure~\ref{fig:tradeoffs} illustrates the advancements in the space/query-time tradeoffs over the years.
Until recently, no distance oracles with subquadratic space and polylogarithmic query time were known. 
Cohen-Addad et al.~\cite{DBLP:conf/focs/Cohen-AddadDW17}, inspired by Cabello's use of Voronoi diagrams for the diameter problem in planar graphs~\cite{DBLP:conf/soda/Cabello17}, provided the first such oracle. 
The currently best known tradeoff~\cite{vorexact} is an oracle with  $\cOtilde(n^{3/2}/Q)$ space and $\cOtilde(Q)$ query-time for any $Q \in [1,n^{1/2}]$.
Note that all known oracles with nearly linear (i.e. $\cOtilde(n)$) space require $\Omega(\sqrt{n})$ time to answer queries.

The holy grail in this area is to design an exact distance oracle for planar graphs with both linear space and constant query-time. 
It is not known whether this goal can be achieved. 
We do know, however, (for nearly twenty years now) that 
approximate distance oracles can get very close.
For any fixed $
\epsilon >0$ there are $(1+\epsilon)$-approximate distance oracles 
that occupy nearly-linear space and answer queries in polylogarithmic, or even constant time~\cite{Thorup04,Klein2002,KawarabayashiST13,SommerICALP11,DBLP:conf/isaac/GuX15,Wulff-Nilsen16,DBLP:conf/esa/ChanS17}. 
However, the main question of whether an approximate answer is the best one can hope for, or whether exact distance oracles for planar graphs with linear space and constant query time exist, remained a wide open important and interesting problem.

\paragraph{Our results and techniques.}
In this paper we approach the optimal tradeoff between space and query time for reporting exact distances.We design exact distance oracles that require almost-linear space and answer distance queries in polylogarithmic time.
Specifically, given a planar graph of size $n$, we show how to construct in roughly $\cO(n^{3/2})$ time a distance oracle admitting any of the following $\langle$space, query-time$\rangle$ tradeoffs (see Theorem~\ref{main} and \cref{cor:tradeoffs} for the exact statements).
\begin{enumerate}[label=(\roman*)]
\item $\langle \cOtilde(n^{1+\epsilon}),\cO(\log^{1/\epsilon} n)\rangle$, for any constant $1/2 \geq \epsilon>0$;
\item $\langle \cO(n \log^{2+1/\epsilon} n),\cOtilde(n^{\epsilon})\rangle$, for any constant $\epsilon>0$;
\item $\langle n^{1+o(1)},n^{o(1)}\rangle$.
\end{enumerate}

\paragraph{Voronoi diagrams.}
The main tool we use to obtain this result is point location in \emph{Voronoi diagrams on planar graphs}.
The concept of Voronoi diagrams has been used in computational geometry for many years (cf.~\cite{Aurenhammer,BergCKO08}). We consider \emph{graphical} (or network) Voronoi diagrams~\cite{DBLP:journals/ipl/Mehlhorn88,DBLP:journals/networks/Erwig00}.
At a high level, a graphical Voronoi diagram with respect to a set $S$ of sites is a partition of the vertices into $|S|$ parts, called Voronoi cells, where the cell of site $s \in S$ contains all vertices that are closer (in the shortest path metric) to $s$ than to any other site in $S$. 
Graphical Voronoi diagrams have been studied and used quite extensively, most notably in applications in road networks (e.g.,~\cite{DBLP:journals/gis/OkabeSFSO08,DBLP:conf/gis/EppsteinG08,DBLP:journals/tcos/HonidenHSW10,DBLP:conf/dasfaa/DemiryurekS12}). 

Perhaps the most elementary operation on Voronoi diagrams is {\em point location}. Given a point (vertex) $v$, one wishes to efficiently determine the site $s$ such that $v$ belongs to the Voronoi cell of $s$. 
Cohen-Addad et al.~\cite{DBLP:conf/focs/Cohen-AddadDW17}, inspired by Cabello's~\cite{DBLP:conf/soda/Cabello17} breakthrough use of Voronoi diagrams in planar graphs, suggested a way to perform point location that led to the first exact distance oracle for planar graphs with subquadratic space and polylogarithmic query time. 
A simpler and more efficient point location mechanism for Voronoi diagrams in planar graphs was subsequently developed in~\cite{vorexact}. In both oracles, the space efficiency is obtained from the fact that the size of the representation of a Voronoi diagram is proportional to the number of sites and not to the total number of vertices.

To obtain our result, we add two new ingredients to the point location mechanism of~\cite{vorexact}.
The first is the use of what might be called {\em external Voronoi diagrams}. 
Unlike previous constructions, instead of working with Voronoi diagrams for every piece in some partition ($r$-division) of the graph, we work with many overlapping Voronoi diagrams, representing the complements of such pieces. 
This is analogous to the use of external dense distance graphs in~\cite{DBLP:conf/focs/BorradaileSW10,DBLP:conf/soda/MozesS12}. 
This approach alone leads to an oracle with space $\cOtilde(n^{4/3})$ and query time $\cO(\log^2 n)$ (see Section~\ref{sec:43}). 
The obstacle with pushing this approach further is that the point location mechanism consists of auxiliary data for each piece, whose size is proportional to the size of the complement of the piece, which is $\Theta (n)$ rather than the much smaller size of the piece. 
We show that this problem can be mitigated by using recursion, and storing much less auxiliary data at a coarser level of granularity. 
This approach is made possible by a more modular view of the point location mechanism which we present in Section~\ref{sec:prelim}, along with other preliminaries. The proof of our main space/query-time tradeoffs is given in Section~\ref{sec:onehole} and the algorithm to efficiently construct these oracles is given in Section~\ref{sec:constr}.

\section{Preliminaries} \label{sec:prelim}

In this section we review the main techniques required for describing our result. Throughout the paper we consider as input a weighted directed planar graph $G=(V(G),E(G))$, embedded in the plane.
(We use the terms weight and length for edges and paths interchangeably throughout the paper.)
We use $|G|$ to denote the number of vertices in $G$. Since planar graphs are sparse, $|E(G)| = \cO(|G|)$ as well.
The dual of a planar graph $G$ is another planar graph $G^*$ whose vertices correspond to faces of $G$ and
vice versa. Arcs of $G^*$ are in one-to-one correspondence with arcs of $G$; 
there is an arc $e^*$ from vertex $f^*$ to vertex $g^*$ of $G^*$
if and only if the corresponding faces $f$ and $g$ of $G$ are to the left and right of the arc $e$, respectively.

We assume that the input graph has no negative length cycles.
We can transform the graph in a standard~\cite{DBLP:conf/esa/MozesW10} way so that all edge weights are non-negative and distances are preserved.
With another standard transformation we can guarantee, in $\cO(n \frac{\log^2 n}{\log \log n})$ time, that each vertex has constant degree and that the graph is triangulated, while distances are preserved and without increasing asymptotically the size of the graph.
We further assume that shortest paths are unique; this can be ensured in $\cO(n)$ time by a deterministic perturbation of the edge weights~\cite{DBLP:conf/stoc/0001FL18}.
Let $d_G(u,v)$ denote the distance from a vertex $u$ to a vertex $v$ in $G$.

\paragraph{Multiple-source shortest paths.} The multiple-source shortest paths (MSSP) data structure~\cite{MSSP} represents all shortest path trees rooted at the vertices of a single face $f$ in a planar graph using a persistent dynamic tree.
It can be constructed in $\cO(n \log n)$ time, requires $\cO(n\log n)$ space, and can report any distance between a vertex of $f$ and any other vertex in the graph in $\cO(\log n)$ time. 
We note that it can be augmented to also return the first edge of this path within the same complexities.
The authors of~\cite{vorexact} show that it can be further augmented ---within the same complexities--- such that given two vertices $u,v \in G$ and a vertex $x$ of $f$ it can return, in $\cO(\log n)$ time, whether $u$ is an ancestor of $v$ in the shortest path tree rooted at $x$ as well as whether $u$ occurs before $v$ in a preorder traversal of this tree.

\paragraph{Separators and recursive decompositions.}
Miller~\cite{DBLP:conf/stoc/Miller84} showed how to compute, in a biconnected triangulated planar graph with $n$ vertices, a simple cycle of size  $\cO(\sqrt{n})$ that separates the graph into two subgraphs, each with at most $2n/3$ vertices. Simple cycle separators  can be used to recursively separate a planar graph until pieces have constant size.
The authors of~\cite{DBLP:conf/stoc/KleinMS13} show how to obtain a complete recursive decomposition tree $\AG$ of $G$ in $\cO(n)$ time. 
$\AG$ is a binary tree whose nodes correspond to subgraphs of $G$ (pieces), with the root being all of $G$ and the leaves being pieces of constant size.
We identify each piece $P$ with the node representing it in $\AG$. We can thus abuse notation and write $P\in \AG$.
An $r$-division~\cite{DBLP:journals/siamcomp/Frederickson87} of a planar graph, for $r \in [1,n]$, is a decomposition of the graph into $\cO(n/r)$ pieces, each of size $\cO(r)$, such that each piece has $\cO(\sqrt{r})$ boundary vertices, i.e.~vertices that belong to some separator along the recursive decomposition used to obtain $P$.
Another desired property of an $r$-division is that the boundary vertices lie on a constant number of faces of the piece (holes).
For every $r$ larger than some constant, an $r$-division with few holes is represented in the decomposition tree $\AG$ of~\cite{DBLP:conf/stoc/KleinMS13}. In fact, it is not hard to see that if the original graph $G$ is triangulated then all vertices of each hole of a piece are boundary vertices.
Throughout the paper, to avoid confusion, we use ``nodes" when referring to $\AG$ and ``vertices" when referring to $G$.
We denote the boundary vertices of a piece $P$ by $\partial P$. We refer to non-boundary vertices as internal.
We assume for simplicity that each hole is a simple cycle. Non-simple cycles do not pose a significant obstacle, as we discuss  at the end of~\cref{sec:onehole}.

It is shown in~\cite[Theorem 3]{DBLP:conf/stoc/KleinMS13} that, given a geometrically decreasing sequence of numbers $(r_m, r_{m-1}, \ldots, r_1)$, where $r_1$ is a sufficiently large constant, $r_{i+1}/r_i=b$ for all $i$ for some $b>1$, and $r_m=n$, we can obtain the $r_i$-divisions for all $i$ in time $\cO(n)$ in total.
For convenience we define the only piece in the $r_m$ division to be $G$ itself.
These $r$-divisions satisfy the property that a piece in the $r_i$-division is a ---not necessarily strict--- descendant (in $\AG$) of a piece in the $r_j$-division for each $j>i$. This ancestry relation between the pieces of an $r$-division can be captures by a tree $\TG$ called the recursive $r$-division tree.

The boundary vertices of a piece $P \in \AG$ that lie on a hole $h$ of $P$ separate the graph $G$ into two subgraphs $G_1$ and $G_2$ (the cycle is in both subgraphs). One of these two subgraphs is enclosed by the cycle and the other is not. Moreover, $P$ is a subgraph of one of these two subgraphs, say $G_1$. We then call $G_2$ the outside of hole $h$ with respect to $P$ and denote it by $P^{h,out}$.
In the sections where we assume that the boundary vertices of each piece lie on a single hole that is a simple cycle, the outside of this hole with respect to $P$ is $G\setminus(P \setminus \partial P)$ and to simplify notation we denote it by $\out{P}$. 

\paragraph{Additively weighted Voronoi diagrams.}
Let $H$ be a directed planar graph with real edge-lengths, and no negative-length cycles. Assume that all faces of $H$ are triangles except, perhaps, a single face $h$.
Let $S$ be the set of vertices that lie on $h$. The vertices of $S$ are called {\em sites}.
Each site $u\in S$ has a weight $\weight(s) \geq 0$ associated with it. The additively weighted distance between a site $s \in S$ and a vertex $v \in V$, denoted by $d^\weight(s,v)$ is defined as $\weight(s)$ plus the length of the $s$-to-$v$ shortest path in $H$.

\begin{definition}
The additively weighted Voronoi diagram of $(S,\weight)$ ($\VD(S, \weight)$) within $H$ is a partition of $V(H)$ into pairwise disjoint sets, one set $\Vor(s)$ for each site $s \in S$. The set $\Vor(s)$, which is called the Voronoi cell of $s$, contains all vertices in $V(H)$ that are closer (w.r.t. $d^\weight$(. , .)) to $s$ than to any other site in $S$.
\end{definition}

There is a dual representation $\VD^{*}(S,\weight)$ (or simply $\VD^{*}$) of a Voronoi diagram $\VD(S,\weight)$. 
Let $H^*$ be the planar dual of $H$. Let $\VD^*_0$ be the subgraph of $H^*$ consisting of the duals of edges $uv$ of $H$ such that $u$ and $v$ are in different Voronoi cells. Let $\VD_1^*$ be the graph obtained from $\VD^*_0$ by contracting edges incident to degree-2 vertices one after another until no degree-2 vertices remain. 
The vertices of $\VD^*_1$ are called Voronoi vertices. A Voronoi vertex $f^* \neq h^*$ is dual to
a face $f$ such that the vertices of $H$ incident to $f$ belong to three
different Voronoi cells. 
We call such a face {\em trichromatic}.
Each Voronoi vertex $f^*$ stores for each vertex $u$ incident to $f$ the site $s$ such that $u \in \Vor(s)$. 
Note that $h^{*}$ (i.e.~the dual vertex corresponding to the face $h$ to which all the sites are incident) is a Voronoi vertex.
Each face of $\VD_1^{*}$ corresponds to a cell
$\Vor(s)$. Hence there are at most $|S|$ faces in $\VD_1^*$. By sparsity of planar graphs, and by the fact that the minimum degree of a vertex in $\VD_1^*$ is 3, the complexity (i.e., the number of vertices, edges and faces) of $\VD_1^*$ is $\cO(|S|)$. 
Finally, we define $\VD^*$ to be the graph obtained from $\VD_1^{*}$ after replacing the node $h^{*}$ by multiple
copies, one for each occurrence of $h$ as an endpoint of an edge in $\VD_1^*$. 
It was shown in~\cite{vorexact} that $\VD^*$ is a forest, and that, if all vertices of $h$ are sites and if the additive weights are such that each site is in its own nonempty Voronoi cell, then $\VD^*$ is a ternary tree. 
See~\cref{fig:tree} (also used in~\cite{vorexact}) for an illustration. \iffull
\begin{figure}[htpb]
\begin{center}
\includegraphics[width=0.65\textwidth]{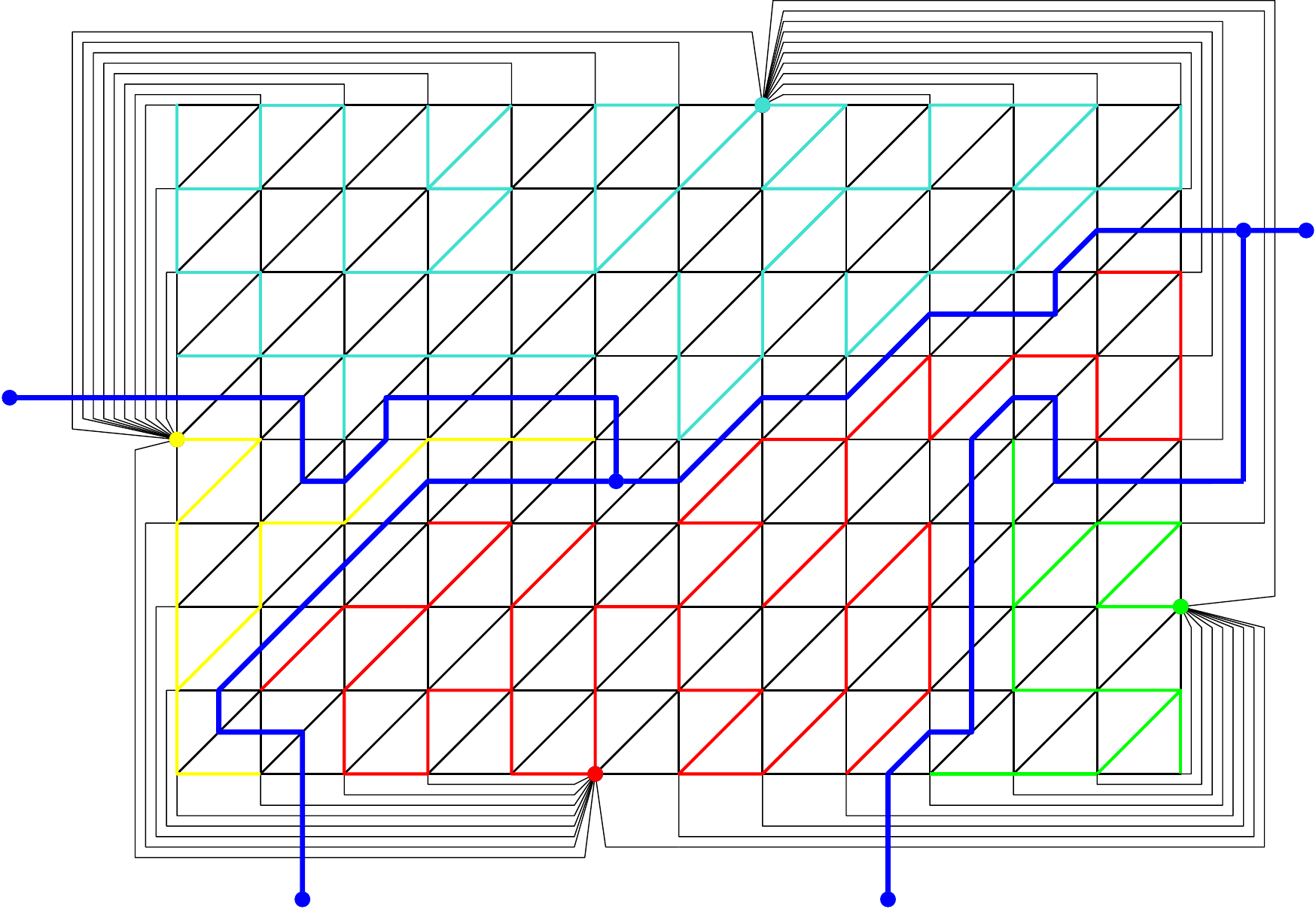}
\end{center}
\caption{A planar graph (black edges) with four sites on the infinite face together with  the dual Voronoi diagram $\VD^{*}$ (in blue). The sites are shown together with their corresponding shortest path trees (in turquoise, red, yellow, and green).\label{fig:tree}}
\end{figure}
\fi

\paragraph{Point location in Voronoi diagrams.}
A \emph{point location query} for a node $v$ in a Voronoi diagram $\VD$ asks for the site $s$ of $\VD$ such that $v \in \Vor(s)$ and for the additive distance from $s$ to $v$. Gawrychowski et al.~\cite{vorexact} described a data structure supporting efficient point location, which is captured by the following theorem.

\begin{theorem}[\cite{vorexact}]\label{thm:pointloc18}
Given an $\cO(|H| \log |H|)$-sized MSSP data structure for $H$ with sources $S$, and after an $\cO(|S|)$-time preprocessing of $\VD^*$, point location queries can be answered in time $\cO(\log^2 |H|)$.
\end{theorem}

We now describe this data structure.
The data structure is essentially the same as in~\cite{vorexact}, but the presentation is a bit more modular. We will later adapt the implementation in~\cref{sec:onehole}.

Recall that $H$ is triangulated (except the face $h$).
For technical reasons, for each face $f\neq h$ of $H$ with vertices $y_1,y_2,y_3$ we embed three new vertices $v^f_1,v^f_2,v^f_3$ inside $f$ and add a directed $0$-length edge from $y_i$ to $v^f_i$.
The main idea is as follows. In order to find the Voronoi cell $\Vor(s)$ to which a query vertex $v$ belongs, it suffices to identify an edge $e^*$ of $\VD^*$ that is adjacent to $\Vor(s)$. Given $e^*$ we can simply check which of its two adjacent cells contains $v$ by comparing the distances from the corresponding two sites to $v$. The point location structure is based on a {\em centroid decomposition} of the tree $\VD^*$ into connected subtrees, and on the ability to determine which of the subtrees is the one that contains the desired edge $e^*$.

The preprocessing consists of just computing a centroid decomposition of $\VD^*$. 
A {\em centroid} of an $n$-node tree $T$ is a node $u\in T$ such that removing $u$ and replacing it with copies, one for each edge incident to $u$, results in a set of trees, each with at most $\frac{n+1}{2}$ edges. A centroid always exists in a tree with at least one edge.
In every step of the centroid decomposition of $\VD^*$, we work with
a connected subtree $T^*$ of $\VD^*$. Recall that 
there are no nodes of degree 2 in $\VD^*$.
If there are
no nodes of degree 3, then $T^*$ consists of a single edge of $\VD^*$, and the decomposition terminates. Otherwise, we choose a centroid $c^*$, and partition $T^*$
into the three subtrees $T^*_{0},T^*_{1},T^*_{2}$
obtained by splitting $c^*$ into three copies, one for each edge incident to $c^*$. Clearly, the depth of the recursive decomposition is $\cO(\log |S|)$. The decomposition can be computed in $\cOtilde(|S|)$ time and can be represented as a ternary tree which we call the {\em centroid decomposition tree}, in $\cO(|S|)$ space. Each non-leaf node of the centroid decomposition tree corresponds to a centroid vertex $c^*$, which is stored explicitly. We will refer to nodes of the centroid decomposition tree by their associated centroid. Each node also corresponds implicitly to the subtree of $\VD^*$ of which $c^*$ is the centroid. The leaves of the centroid decomposition tree correspond to single edges of $\VD^*$, which are stored explicitly.

Point location queries for a vertex $v$ in the Voronoi diagram $\VD$ are answered by employing procedure \textsc{PointLocate}$(\VD^*,v)$ (\cref{alg:pointloc1}), which takes as input the Voronoi diagram, represented by the centroid decomposition of $\VD^*$, and the vertex $v$.
This in turn calls the recursive procedure \textsc{HandleCentroid}$(\VD^*,c^*,v)$ (\cref{alg:handle1}), where $c^*$ is the root centroid in the centroid decomposition tree of $\VD^*$.

\medskip
\alglanguage{pseudocode}
\begin{algorithm}[b]
\caption{$\textsc{PointLocate}$($\VD^*$, $v$)
\label{alg:pointloc1}}
 \textbf{Input:} The centroid decomposition of the dual representation $\VD^*$ of a Voronoi diagram $\VD$ and a vertex $v$.\\
 \textbf{Output:} The site $s$ of $\VD$ such that $v \in \Vor(s)$ and the additive distance from $s$ to $v$.
\begin{algorithmic}[1]
    \State $c^* \leftarrow$ root centroid in the centroid decomposition of $\VD^*$
\State \textbf{return} \textsc{HandleCentroid}$(\VD^*,c^*,v)$

\Statex
\end{algorithmic}
  \vspace{-0.4cm}%
\end{algorithm}

We now describe the procedure \textsc{HandleCentroid}. It gets as input the Voronoi diagram, represented by its centroid decomposition tree, the centroid node $c^*$ in the centroid decomposition tree that should be processed, and the vertex $v$ to be located. It returns the site $s$ such that $v \in \Vor(s)$, and the additive distance to $v$.
If $c^*$ is a leaf of the centroid decomposition, then its corresponding subtree of $\VD^*$ is the single edge $e^*$ we are looking for (Lines~\ref{ps1:basestart}-\ref{ps1:baseend}). 
Otherwise, $c^*$ is a non-leaf node of the centroid decomposition tree, so it corresponds to a node in the tree $\VD^*$, which is also a vertex of the dual $H^*$ of $H$. Thus, $c$ is a face of $H$. Let the vertices of $c$ be $y_1, y_2$ and $y_3$.
We obtain $s_i$, the site such that $y_i \in \Vor(s_i)$, from the representation of $\VD^*$ (Line~\ref{ps1:sites}).
(Recall that $s_i \neq s_k$ if $i \neq k$ since $c$ is a trichromatic face.)
Next, for each $i$, we retrieve the additive distance from $s_i$ to $v$. 
Let $s_j$ be the site among them with minimum additive distance to $v$.

If $v$ is an ancestor of the node $y_j$ in the shortest path tree rooted at $s_j$, then $v \in \Vor(s_j)$ and we are done. Otherwise, handling $c$ consists of finding the child $c'^*$ of $c^*$ in the centroid decomposition tree whose corresponding subtree contains $e^*$. 
It is shown in~\cite{vorexact} that identifying $c'^*$ amounts to determining a certain left/right relationship between $v$ and $v^c_j$. (Recall that $v^c_j$ is the artificial vertex 
embedded inside the face $c$ with an incoming 0-length edge from $y_j$.) We next make this notion more precise.

\begin{definition}
For a tree $T$ and two nodes $x$ and $y$, such that none is an ancestor of the other, we say that $x$ is {\em left of} $y$ if the preorder number of $x$ is smaller than the preorder number of $y$. Otherwise, we say that $x$ is {\em right of} $y$.
\end{definition}

It is shown in~\cite[Section 4]{vorexact} that given the left/right/ancestor relationship of $v$ and $v_j^c$, one can determine, in constant time, the child of $c^*$ in the centroid decomposition tree containing $e^*$. We call the procedure that does this \textsc{NextCentroid}. Having found the next centroid in Line~\ref{ps1:next}, \textsc{HandleCentroid} moves on to handle it recursively.  

\medskip
\alglanguage{pseudocode}
\begin{algorithm}[t]
\caption{\textsc{HandleCentroid}$(\VD^*, c^*, v)$}\label{alg:handle1}
 \textbf{Input:} The centroid decomposition of the dual representation $\VD^*$ of a Voronoi diagram, a centroid $c^*$ and a vertex $v$ that belongs to the Voronoi cell of some descendant of $c^*$ in the centroid decomposition of $\VD^*$.\\
 \textbf{Output:} The site $s$ of $\VD$ such that $v \in \Vor(s)$ and the additive distance from $s$ to $v$. 
\begin{algorithmic}[1]
\If { $c^*$ is a leaf}\label{ps1:basestart}
	\State $s_1,s_2 \leftarrow$ sites corresponding to $c^*$
    \For{$k=1,2$}
    	\State $d_k \leftarrow \weight(s_k)+d_H(s_k,v)$\label{ps1:dist1} 
    \EndFor
    \State $j \leftarrow \argmin_k(d_k)$
    \State \textbf{return} $(s_j,d_j)$
\EndIf\label{ps1:baseend}
\State $s_1,s_2,s_3 \leftarrow$ sites corresponding to $c^*$ \label{ps1:sites}
\For{$k=1,2,3$}
	\State $d_k \leftarrow \weight(s_k)+d_H(s_k,v)$ \label{ps1:dist2}
\EndFor
\State $j \leftarrow \argmin_k(d_k)$
\State $a \leftarrow$ left/right/ancestor relationship of $v$ to $v_j^c$ in the shortest path tree of $s_j$ in H \label{ps1:leftright}
\If{$v$ is an ancestor of $v_j^c$}
	\State \textbf{return} $(s_j,d_j)$
\Else
	\State $c \leftarrow \textsc{NextCentroid}(c,a)$ \label{ps1:next}
	\State \textbf{return} \textsc{HandleCentroid}$(\VD^*, c , v)$
\EndIf
\Statex
\end{algorithmic}
  \vspace{-0.4cm}
\end{algorithm}

The efficiency of procedure \textsc{HandleCentroid} depends on the time to compute distances in $H$ (Lines~\ref{ps1:dist1} and~\ref{ps1:dist2}) and the left/right/ancestor relationship (Line~\ref{ps1:leftright}).
Given an MSSP data structure for $H$, with sources $S$, each of these operations can be performed in time $\cO(\log|H|)$ and hence~\cref{thm:pointloc18} follows.

\iffull
\section{An $\cOtilde(n^{4/3})$-space $\cO(\log^2 n)$-query oracle}\label{sec:43}

We state the following result from~\cite{vorexact} that we use in the oracle presented in this section.

\begin{theorem}[\cite{vorexact}]\label{thm:vorvan}
For a planar graph $G$ of size $n$, there is an $\cO(n^{3/2})$-sized data structure that answers distance queries in time $\cO(\log n)$.
\end{theorem}

For clarity of presentation, we first describe our oracle under the assumption that the boundary vertices of each piece $P$ in the $r$-division of the graph lie on a single hole and that each such hole is a simple cycle.
Multiple holes and non-simple cycles do not pose any significant complications; we explain how to treat pieces with multiple holes that are not necessarily simple cycles, separately.

\paragraph{Data Structure.} We obtain an $r$-division for $r=n^{2/3} \sqrt{\log n}$. The data structure consists of the following for each piece $P$ of the $r$-division:
\begin{enumerate}
\item The $\cO(|P|^{3/2})$-space, $\cO(\log |P|)$-query-time distance oracle of~\cref{thm:vorvan}. These occupy $\cO(n\sqrt{r})$ space in total.
\item Two MSSP data structures, one for $P$ and one for $\out{P}$, both with sources the nodes of $\partial P$.
The MSSP data structure for $P$ requires space $\cO(r \log r)$, while the one for $\out{P}$ requires space $\cO(n \log n)$.
The total space required for the MSSP data structures is $\cO(\frac{n^2}{r}\log n)$, since there are $\cO(\frac{n}{r})$ pieces.
\item For each node $u$ of $P$:
\begin{itemize}
\item $\VDin^*(u,P)$, the dual representation of the Voronoi diagram for $P$ with sites the nodes of $\partial P$, and additive weights the distances from $u$ to these nodes in $G$;
\item $\VDout^*(u,P)$, the dual representation of the Voronoi diagram for $\out{P}$ with sites the nodes of $\partial P$, and additive weights the distances from $u$ to these nodes in $G$.
\end{itemize}
The representation of each Voronoi diagram occupies $\cO(\sqrt{r})$ space and hence, since each vertex belongs to a constant number of pieces, all Voronoi diagrams require space $\cO(n \sqrt{r})$. 
\end{enumerate}

\paragraph{Query.}
We obtain a piece $P$ of the $r$-division that contains $u$.
Let us first suppose that $v \in P$. We have to consider both the case that the shortest $u$-to-$v$ path crosses $\partial P$ and the case that it does not. If it does cross, we retrieve this distance by performing a point location query for $v$ in the Voronoi diagram $\VDin(u,P)$. If the shortest $u$-to-$v$ path does not cross $\partial P$, the path lies entirely within $P$. We thus retrieve the distance by querying the exact distance oracle of~\cref{thm:vorvan} stored for $P$. The answer is the minimum of the two returned distances. This requires $\cO(\log^2 n)$ time by~\cref{thm:pointloc18,thm:vorvan}.
Else, $v \not\in P$ and the shortest path from $u$ to $v$ must cross $\partial P$. The answer can be thus obtained by a point location query for $v$ in the Voronoi diagram $\VDout(u,P)$ in time $\cO(\log^2 n)$ by~\cref{thm:pointloc18}. The pseudocode of the query algorithm is presented below as procedure \textsc{SimpleDist}$(u,v)$ (\cref{alg:dist1}).

\medskip
\alglanguage{pseudocode}
\begin{algorithm}[H]
\caption{\textsc{SimpleDist}$(u,v)$}\label{alg:dist1}
\textbf{Input:} Two nodes $u$ and $v$.\\
\textbf{Output:} $d_G(u,v)$.
\begin{algorithmic}[1]
\State $P \leftarrow$ a piece of the $r$-division containing $u$
\If{$v\in P$}
	\State $d_1 \leftarrow d_{P}(u,v)$
    \State $d_2 \leftarrow \textsc{PointLocate}(\VD_{in}^*(u,P),v)$
    \State \textbf{return} $\min(d_1,d_2)$
\Else
	\State \textbf{return} \textsc{PointLocate}$(\VD_{out}^*(u,P),v)$
\EndIf
\Statex
\end{algorithmic}
  \vspace{-0.4cm}
\end{algorithm}

\paragraph{Dealing with holes.}
The data structure has to be modified as follows.
\begin{enumerate}
\setcounter{enumi}{1}
\item For each hole $h$ of $P$, two MSSP data structures, one for $P$ and one for $P^{h,out}$, both with sources the nodes of $\partial P$  that lie on $h$.
\item For each node $u$ of $P$, for each hole $h$ of $P$:
\begin{itemize}
\item $\VDin^*(u,P,h)$, the dual representation of the Voronoi diagram for $P$ with sites the nodes of $\partial P$ that lie on $h$, and additive weights the distances from $u$ to these nodes in $G$;
\item $\VDout^*(u,P,h)$, the dual representation of the Voronoi diagram for $P^{h,out}$ with sites the nodes of $\partial P$ that lie on $h$, and additive weights the distances from $u$ to these nodes in $G$.
\end{itemize}
\end{enumerate}

As for the query, if $v \in P$ we have to perform a point location query in $\VDin(u,P,h)$ for each hole $h$ of $P$. Else $v \not\in P$ and we have to perform a point location query in $\VDout(u,P,h)$ for the hole $h$ of $P$ such that $v \in P^{h,out}$. We can afford to store the required information to identify this hole explicitly in balanced search trees.

We thus obtain the following result.

\begin{theorem}
For a planar graph $G$ of size $n$, there is an $\cO(n^{4/3}\sqrt{\log n})$-sized data structure that answers distance queries in time $\cO(\log^2 n)$.
\end{theorem}
\fi

\section{An oracle with almost optimal tradeoffs}\label{sec:onehole}
\iffull
In this section we describe how to obtain an oracle with almost optimal space to query-time tradeoffs. 
Consider the oracle in the previous section. The size of the representation we store for each Voronoi diagram is proportional to the number of sites, while the size of an MSSP data structure is roughly proportional to the size of the graph in which the Voronoi diagram is defined. 
Thus, the MSSP data structures that we store for the outside of pieces of the $r$-division are the reason that the oracle in the previous section requires $\Omega(n^{4/3})$ space. However, storing the Voronoi diagrams for the outside of each piece is not a problem.
For instance, if we could somehow afford to store MSSP data structures for $P$ and $\out{P}$ of each piece $P$ of an $n^\epsilon$-division using just $\cO(n^\epsilon)$ space, then plugging $r=n^\epsilon$ into the data structure from the previous sections would yield an oracle with space $\cO(n^{1+\epsilon})$ and query-time $\cO(\log^2 n)$.

We cannot hope to have such a compact MSSP representation. However, 
we can use recursion to get around this difficulty. 
\fi
We compute a recursive $r$-division, represented by a recursive $r$-division tree $\TG$. 
We store, for each piece $P \in \TG$, the Voronoi diagram for $\out{P}$. 
However, instead of storing the costly MSSP for the entire $\out{P}$, we store the MSSP data structure (and some additional information) just for the portion of $\out{P}$ that belongs to the parent $Q$ of $P$ in $\TG$. 
Roughly speaking, when we need to perform point location on a vertex of $\out{P}$ that belongs to $Q$ we can use this MSSP information.
When we need to perform point location on a vertex of $\out{P}$ that does not belong to $Q$ (i.e., it is also in $\out Q$), we recursively invoke the point location mechanism for $\out{Q}$. 

We next describe the details. For clarity of presentation, we assume that the boundary vertices of each piece $P$ in $\TG$ lie on a single hole which is a simple cycle.
We later explain how to remove these assumptions.
In what follows, if a vertex in a Voronoi diagram can be assigned to more than one Voronoi cell, we assign it to the Voronoi cell of the site with largest additive weight. In other words, since we define the additive weights as distances from some vertex $u$, and since shortest paths are unique, we assign each vertex $v$ to the Voronoi cell of the last site on the shortest $u$-to-$v$ path. 
In particular, this implies that there are no empty Voronoi cells as every site belongs to its own cell. Thus $\VD^*$ is a ternary tree (see~\cite{vorexact}). We can make such an assignment by perturbing the additive weights to slightly favor sites with larger distances from $u$ at the time that the Voronoi diagram is constructed.

\paragraph{The data structure.}
Consider a recursive $(r_{m},\ldots,r_1)$-division of $G$ for some $n=r_{m}> \cdots > r_1=\cO(1)$ to be specified later.
Recall that our convention is that the $r_m$-division consists of $G$ itself.
For convenience, we define each vertex $v$ to be a boundary vertex of a singleton piece at level~$0$ of the recursive division. Denote the set of pieces of the $r_i$-division by $\mathcal{R}_i$. Let $\TG$ denote the tree representing this recursive $(r_m,\ldots,r_0)$-division (each singleton piece $v$ at level $0$ is attached as the child of an arbitrary piece $P$ at level $1$ such that $v \in P$).

We will handle distance queries between vertices $u,v$ that have the same parent $P$ in $\mathcal{R}_{1}$ separately by storing these distances explicitly (this takes $\cO(n)$ space because pieces at level $1$ have constant size).

\renewcommand{\VDin}{\textsf{VD}_{near}}
\renewcommand{\VDout}{\textsf{VD}_{far}}

The oracle consists of the following for each $0 \leq i \leq m-1$, for each piece $R \in \mathcal{R}_i$ whose parent in $\TG$ is $Q \in \mathcal{R}_{i+1}$:
\begin{enumerate}
\item If $i>0$, two MSSP data structures for $Q \setminus (R \setminus \partial R)$, with sources $\partial R$ and $\partial Q$, respectively.
\item If $i<m-1$, for each boundary vertex $u$ of $R$:
\begin{itemize}[leftmargin=3pt]
\item $\VDin^*(u,Q)$, the dual representation of the Voronoi diagram for $Q \setminus (R \setminus \partial R)$ with sites the nodes of $\partial Q$, and additive weights the distances from $u$ to these nodes in $\out{R}$;
\item $\VDout^*(u,Q)$, the dual representation of the Voronoi diagram for $\out{Q}$ with sites the nodes of $\partial Q$, and additive weights the distances from $u$ to these nodes in $\out{R}$;
\item if $i>0$, the coarse tree $T_u^R$, which is the tree obtained from the shortest path tree rooted at $u$ in $\out{R}$ (the fine tree) by  contracting any edge incident to a vertex not in $\partial Q \cup \partial R$. 
Note that the left/right/ancestor relationship between vertices of $\partial Q \cup \partial R$ in the coarse tree is the same as in the fine tree.
Also note that each (coarse) edge in $T_u^R$ originates from a contracted path in the fine tree. 
We preprocess the coarse tree $T_u^R$ in time proportional to its size to allow for $\cO(1)$-time lowest common ancestor (LCA)\footnote{An LCA query takes as input two nodes of a rooted tree and returns the deepest node of the tree that is an ancestor of both} and level ancestor\footnote{A level ancestor query takes as input a node $v$ at depth $d$ and an integer $\ell\leq d$ and returns the ancestor of $v$ that is at depth $\ell$.} queries~\cite{DBLP:conf/latin/BenderF00,DBLP:journals/tcs/BenderF04}. 
In addition, for every (coarse) edge of $T_u^R$, we store the first and last edges of the underlying path in the fine tree. We also store the preorder numbers of the vertices of $T_u^R$.
\end{itemize}
\end{enumerate}

\noindent \textbf{Space.} The space required to store the described data structure is $\cO(n \sum_i\frac{r_{i+1}}{r_i} \log r_{i+1})$.
Part 1 of the data structure: At the $r_i$-division, we have $\cO(\frac{n}{r_i})$ pieces and for each of them, we store two MSSP data structures, each of size $\cO(r_{i+1} \log r_{i+1})$.
Thus, the total space required for the MSSP data structures is $\cO(n \sum_i\frac{r_{i+1}}{r_i} \log r_{i+1})$.
Part 2 of the data structure: At the $r_i$-division, we store, for each of the $\cO(\frac{n}{\sqrt{r_i}})$ boundary nodes, the representation of two Voronoi diagrams and a coarse tree, each of size $\cO(\sqrt{r_{i+1}})$. The space for this part is thus 
$\cO(n \sum_i ({\frac{r_{i+1}}{r_i}})^{1/2})$.

\paragraph{Query.}

Upon a query $d_G(u,v)$ for the distance between vertices $u$ and $v$ in $G$,
we pick the singleton piece $Q_0 =\{u\} \in \mathcal{R}_0$ and call the  
procedure \textsc{Dist}$(u,v,Q_0)$ presented below. In what follows we denote by $Q_j$ the ancestor of $Q_0$ in $\TG$ that is in $\mathcal{R}_j$. 

The procedure \textsc{Dist}$(u,v,Q_i)$ (\cref{alg:dist2}) gets as input a piece $Q_i$, a source vertex $u\in \partial Q_i$ and a destination vertex $v \in \out{Q_i}$.
The output of \textsc{Dist}$(u,v,Q_i)$ is a tuple $(\beta,d)$, where $d$ is the distance from $u$ to $v$ in $\out{Q_i}$ and $\beta$ is a sequence of certain boundary vertices of the recursive division that occur along the recursive query. 
(Note that $\out{Q_0}=G$ and thus \textsc{Dist}$(u,v,Q_0)$ returns $d_G(u,v)$.)
Let~$i^*$ be the smallest index such that $v \in Q_{i^*}$. For every $1 \leq i < i^*$, the $i$'th element in the list $\beta$ is the last boundary vertex on the shortest $u$-to-$v$ path that belongs to $\partial Q_i$.
As we shall see, the list $\beta$ enables us to obtain the information that enables us to determine the required left/right/ancestor relationships during the recursion.
\textsc{Dist}$(u,v,Q_i)$ works as follows.
\begin{enumerate}
\item If $v \in Q_{i+1} \setminus Q_i$, it computes the required distance by considering the minimum of the distance from $u$ to $v$ returned by a query in the MSSP data structure for $Q_{i+1}\setminus(Q_i\setminus \partial Q_i)$ with sources $\partial Q_i$ and the distance returned by a (vanilla) point location query for $v$ in $\VDin(u, Q_{i+1})$ using the procedure \textsc{PointLocate}. The first query covers the case where the shortest path from $u$ to $v$ lies entirely within $Q_{i+1}$, while the second one covers the complementary case. The time required in this case is $\cO(\log^2 n)$.
\item If $v \not\in Q_{i+1}$, \textsc{Dist} performs a recursive point location query on $\VDout(u, Q_{i+1})$ by calling the procedure \textsc{ModifiedPointLocate}. 
\end{enumerate}

\medskip
\alglanguage{pseudocode}
\begin{algorithm}[h]
\caption{\textsc{Dist}$(u,v,Q_i)$}\label{alg:dist2}
\textbf{Input:} A piece $Q_i$, a vertex $u \in \partial Q_i$ and a vertex $v \in \out{Q_i}$ ($v$ may belong to $Q_i$ if $i=m$).\\
\textbf{Output:} $(s,d)$, where $s$ is a list of (some) boundary vertices on the $u$-to-$v$ shortest path in $\out{Q_i}$, and $d$ is the $u$-to-$v$ distance in $\out{Q_i}$. 
\begin{algorithmic}[1]
\If{$v\in Q_{i+1}$}\label{psdist2:start1} 
	\State $(\beta,d) \leftarrow  (null,$ distance returned by querying the MSSP data structure for $Q_{i+1}\setminus(Q_i\setminus \partial Q))$
	\State $(\beta',d') \leftarrow \textsc{PointLocate}(\VDin^*(u,Q_{i+1}),v)$ \label{psdist2:near}
	\If{$d'<d$}
		\State $(\beta,d) \leftarrow (\beta',d')$
	\EndIf
\ElsIf{$i=m-2$}
	\State $(\beta,d) \leftarrow \textsc{PointLocate}(\VDout^*(u,Q_{m-1}),v)$\label{psdist2:end1}
\Else
	\State $(\beta,d) \leftarrow \textsc{ModifiedPointLocate}(\VDout^*(u,Q_{i+1}),v,i+1)$\label{psdist2:modpointloc} 
\EndIf
\State \textbf{return} $(\beta,d)$
\Statex
\end{algorithmic}
  \vspace{-0.4cm}
\end{algorithm}

Since \textsc{Dist} returns a list of boundary vertices as well as the distance, \textsc{PointLocate} and \textsc{HandleCentroid} must pass and augment these lists as well. The pseudocode for \textsc{ModifiedPointLocate} is identical to that of \textsc{PointLocate}, except it calls \textsc{ModifiedHandleCentroid} instead of \textsc{HandleCentroid}; see~\cref{alg:pointloc2}.
In what follows, when any of the three procedures is called with respect to a piece $Q_i$, we refer to this as an invocation of this procedure at level $i$.

\medskip
\alglanguage{pseudocode}
\begin{algorithm}[t]
\caption{\textsc{ModifiedPointLocate}$(\VD^*, v, i)$}\label{alg:pointloc2}
\textbf{Input:} The centroid decomposition of the dual representation $\VD^*$ of a Voronoi diagram $\VD$ for $\out{Q_i}$ with sites $\partial Q_i$ and a vertex $v \in \out{Q_i}$.\\
\textbf{Output:} The site $s$ of $\VD$ such that $v \in \Vor(s)$ and the additive distance from $s$ to $v$. 
\begin{algorithmic}[1]
    \State $c^* \leftarrow$ root centroid in the centroid decomposition of $\VD^*$
\State \textbf{return} {\color{red}\textsc{ModifiedHandleCentroid}$(\VD^*,c^*,v, i)$}
\Statex
\end{algorithmic}
  \vspace{-0.4cm}%
\end{algorithm}

The pseudocode of \textsc{ModifiedHandleCentroid} (\cref{alg:handle2}) is similar to that of the procedure $\textsc{HandleCentroid}$, except that when the site $s$ such that $v \in \Vor(s)$ is found, $s$ is prepended to the list. 
A more significant change in \textsc{ModifiedHandleCentroid} stems from the fact that, since we no longer have an MSSP data structure for all of $\out{Q_i}$, 
we use recursive calls to \textsc{Dist} to obtain the distances from the sites $s_k$ to $v$ in $\out{Q_{i}}$ in Lines~\ref{ps12:dist1} and~\ref{ps12:dist2}.
We highlight these changes in red in the pseudocode provided below (\cref{alg:handle2}).
We next discuss how to determine the required left/right/ancestor relationships in $\VDout$ (Line~\ref{ps12:leftright}) in the absence of MSSP information for the entire $\out{Q_i}$.

\medskip
\alglanguage{pseudocode}
\begin{algorithm}[h]
\caption{\textsc{ModifiedHandleCentroid}$(\VD^*, c^*, v, i)$}\label{alg:handle2}
\textbf{Input:} The centroid decomposition of the dual representation $\VD^*$ of a Voronoi diagram $\VD$ for $\out{Q_i}$ with sites $\partial Q_i$, a centroid $c^*$ and a vertex $v \in \out{Q_i}$ that belongs to the Voronoi cell of some descendant of $c^*$ in the centroid decomposition of $\VD^*$.\\
\textbf{Output:} The site $s$ of $\VD$ such that $v \in \Vor(s)$ and the additive distance from $s$ to $v$. 
\begin{algorithmic}[1]
\If { $c^*$ is a leaf}
	\State $s_1,s_2 \leftarrow$ sites corresponding to $c^*$
    \For{$k=1,2$}
    	\State $({\color{red}\beta_k},d_k) \leftarrow {\color{red}\textsc{Dist}}(s_k,v,{\color{red}Q_{i}})$\label{ps12:dist1} 
    \EndFor
    \State {\color{black}$j \leftarrow \argmin_k(\weight(s_k)+ d_k)$}
    \State \textbf{return} $({\color{red}prepend(s_j,\beta_j)},\weight(s_j)+d_j)$ \label{ps12:prepend1}
\EndIf
\State $s_1,s_2,s_3 \leftarrow$ sites corresponding to $c^*$ \label{ps12:3sitesstart}
\For{$k=1,2,3$}
	\State $({\color{red}\beta_k},d_k) \leftarrow {\color{red}\textsc{Dist}}(s_k,v,{\color{red}Q_{i}})$
	\label{ps12:dist2}
\EndFor
\State {\color{black}$j \leftarrow \argmin_k(\weight(s_k)+d_k)$}\label{ps12:3sitesend}
\State $a \leftarrow$ left/right/ancestor relationship of $v$ to $v_j^c$ in the shortest path tree of $s_j$ in $\out{Q_i}$ \label{ps12:leftright}
\If{$v$ is an ancestor of $v_j^c$}
	\State \textbf{return} $({\color{red}prepend(s_j,\beta_j)},\weight(s_j)+d_j)$ \label{ps12:prepend2}
\Else
	\State $c \leftarrow \textsc{NextCentroid}(c,a)$
	\State \textbf{return} {\color{red}\textsc{ModifiedHandleCentroid}}$(\VD^*, c , v, {\color{red}i})$
\EndIf
\Statex
\end{algorithmic}
  \vspace{-0.4cm}%
\end{algorithm}

\paragraph{Left/right/ancestor relationships in $\VDout$.}
Let $s_j$ be the site among the three sites corresponding to a centroid $c^*$ such that $v$ is closest to $s_j$ with respect to the additive distances (see Lines~\ref{ps12:3sitesstart}-\ref{ps12:3sitesend} of~\cref{alg:handle2}). Recall that $y_j$ is the vertex of the centroid face $c$ that belongs to $\Vor(s_j)$ and that $v_j^c$ is the artificial vertex connected to $y_j$ and embedded inside the face $c$.
In Line~\ref{ps12:leftright} of \textsc{ModifiedHandleCentroid} we have to determine whether $v$ is an ancestor of $y_j$ in the shortest path tree rooted at $s_j$ in $\out{Q_{i}}$, and if not, whether the $s_j$-to-$v$ path is left or right of the $s_j$-to-$v_j^c$ path.
To avoid clutter, we will omit the subscript $j$ in the following, and refer to $s_j$ as $s$, to $v_j^c$ as $v^c$, and to the sequence of boundary vertices $\beta_j$ returned by the recursion as $\beta$. 
To infer the relationship between the two paths we use the sites (boundary vertices) stored in the list $\beta$. We prepend $s$ to $\beta$, and, if $v$ is not already the last element of $\beta$, we append $v$ to $\beta$ and use a flag to denote that $v$ was appended.

To be able to compare the $s$-to-$v$ path with the $s$-to-$v^c$ path, we perform another recursive call \textsc{Dist}$(s,v^c,Q_i)$ (this call is implicit in Line~\ref{ps12:leftright}). 
Let $\gamma$ be the list of sites returned by this call. As above, we prepend $s$ to $\gamma$, and append $v^c$ if it is not already the last element. 
The intuition is that the lists $\beta$ and $\gamma$ are a coarse representation of the $s$-to-$v$ and $s$-to-$v^c$ shortest paths, and that we can use this coarse representation to roughly locate the vertex where the two paths diverge. 
The left/right relationship between the two paths is determined solely by the left/right relationship between the two paths at that vertex (the divergence point).
We can use the local coarse tree information or the local MSSP information to infer this  relationship. 

Recall that $s$ is a boundary vertex of the piece $Q_i$ at level $i$ of the recursive division. More generally, for any $k \geq 0$, $\beta[k]$ is a boundary vertex of $Q_{i+k}$ (except, possibly, when $\beta[k]=v$). 
To avoid this shift between the index of a site in the list and the level of the corresponding piece in the recursive $r$-division, we prepend $i$ empty elements to both lists, so that now, for any $k \geq i$, $\beta[k]$ is a boundary vertex of $Q_{k}$. 
Let $k$ be the largest integer such that $\beta[k] = \gamma[k]$. Note that $k$ exists because $s$ is the first vertex in both $\beta$ and $\gamma$.

\begin{observation}
The restriction of the shortest path from $\beta[i]$ to $\beta[i+1]$ in $\out{Q_i}$ to the nodes of $\partial Q_i \cup \partial Q_{i+1}$ is identical to the path from $\beta[i]$ to $\beta[i+1]$ in $T^{Q_i}_{\beta[i]}$.
\end{observation}

We next analyze the different cases that might arise  and describe how to correctly infer the left/right relationship in each case. An illustration is provided as~\cref{fig:cases}.

\begin{enumerate}
\item $(\beta[k+1]=v \notin \partial Q_{k+1})$ or $(\gamma[k+1]=c \notin \partial Q_{k+1})$. This corresponds to the case that in at least one of the lists there are no boundary vertices after the divergence point. We have a few cases depending on whether $\beta[k+1]$ and $\gamma[k+1]$ are in $Q_{k}$ or not.
\begin{enumerate}[label={\alph*}.]
\item If $\beta[k+1],\gamma[k+1] \in Q_k$, then it must be that $\beta[k+1]=v$ and $\gamma[k+1]=c$ were appended to the sequences manually. In this case we can query the MSSP data structure for $Q_{k} \setminus (Q_{k-1} \setminus \partial Q_{k-1})$ with sources the boundary vertices of $\partial Q_k$ to determine the left/right/ancestor relationship.
\item If $\beta[k+1],\gamma[k+1] \not\in Q_k$, we can use the MSSP data structure stored for $Q_{k+1} \setminus (Q_{k} \setminus \partial Q_{k})$ with sources the nodes of $\partial Q_{k}$. Let us assume without loss of generality that $\beta[k+1]=v \notin \partial Q_{k+1}$. We then define $w$ to be $\gamma[k+1]$ if $\gamma[k+1]=c \not\in \partial Q_{k+1}$ or the child $w$ of the root $\beta[k]$ of the coarse $T^{Q_k}_{\beta[k]}$ that is an ancestor of $\gamma[k+1]$ otherwise. (In the latter case we can compute $w$ in $\cO(1)$ time with a level ancestor query.) We then query the MSSP data structure for the relation between $v$ and the $w$. Note that the relation between the $s$-to-$v$ and $s$-to-$c$ paths is the same as the relation between the $s$-to-$v$ and $s$-to-$w$ paths. 
\item Else, one of $\beta[k+1],\gamma[k+1]$ is in $Q_k$ and the other is not. Let us assume without loss of generality that $\beta[k+1]=v \in Q_k$. We can infer the left/right relation by looking at the circular order of the following edges: (i) the last fine edge in the root-to-$\beta[k]$ path in $T^{Q_{k-1}}_{\beta[k-1]}$, (ii) the first edge in the $\beta[k]$-to-$\beta[k+1]$ path, and (iii) the first edge in the $\beta[k]$-to-$w$ path, where $w$ is defined as in Case 1(b). Edge (i) is stored in $T^{Q_{k-1}}_{\beta[k-1]}$ ---note that $\beta[k-1]$ exists since $\beta[k+1] \in Q_k$. 
Edge (ii) can be retrieved from the MSSP data structure for $Q_k \setminus (Q_{k-1} \setminus \partial Q_{k-1})$ with sources the boundary vertices of $\partial Q_k$, and edge (iii) can be retrieved from the MSSP data structure for $Q_{k+1} \setminus (Q_{k} \setminus \partial Q_{k})$ with sources the boundary vertices of $\partial Q_k$.
\end{enumerate}
\item $\beta[k+1],\gamma[k+1] \in \partial Q_{k+1}$.
We first compute the LCA $x$ of $\beta[k+1]$ and $\gamma[k+1]$ in $T^{Q_k}_{\beta[k]}$.
\begin{enumerate}[label={\alph*}.]
\item If neither of $\beta[k+1],\gamma[k+1]$ is an ancestor of the other in $T^{Q_k}_{\beta[k]}$ we are done by utilising the preorder numbers stored in $T^{Q_k}_{\beta[k]}$.
\item Else, $x$ is one of $\beta[k+1],\gamma[k+1]$.
We can assume without loss of generality that $x=\beta[k+1]$. 
Using a level ancestor query we find the child $z$ of $\beta[k+1]$ in $T^{Q_{k}}_{\beta[k]}$ that is a (not necessarily strict) ancestor of $\gamma[k+1]$. 
The $x$-to-$z$ shortest path in $\out{Q_k}$ is internally disjoint from $\partial Q_{k+1}$; i.e.~it starts and ends at $\partial Q_{k+1}$, but either entirely lies inside $Q_{k+1}$ or entirely outside $Q_{k+1}$.
If $\beta[k+2]=v\notin \partial Q_{k+2}$ let $w$ be $v$; otherwise let $w$ be the child of the root in $T^{Q_{k+1}}_{\beta[k+1]}$ that is an ancestor of $\beta[k+2]$. The $x$-to-$w$ shortest path also lies either entirely inside or entirely outside $Q_{k+1}$. It remains to determine the relationship of the $x$-to-$w$ and the $x$-to-$z$ paths. 
\begin{enumerate}[label={\roman*}.]
\item If the shortest $x$-to-$z$ path and the shortest $x$-to-$w$ path lie on different sides of $\partial Q_{k+1}$, we can infer the left/right relation by looking at the circular order of the following edges: (i) the last edge on the coarse edge from $\beta[k]$ to $x$ in $T^{Q_k}_{\beta[k]}$, (ii) the first edge on the $x$-to-$z$ shortest path in $\out{Q_{k}}$, and (iii) the first edge on the $x$-to-$w$ shortest path in $\out{Q_{k}}$. Edges (ii) and (iii) can be retrieved from the MSSP data structure.
\item Else, the $x$-to-$z$ and the $x$-to-$w$ shortest paths lie on the same side, and we compute the required relation using the MSSP data structure for this side, with sources the sites of $\partial Q_{k+1}$.
\end{enumerate}
\end{enumerate}
\end{enumerate}

\begin{figure}
\centering
\subfloat[Case 1.a.]{
{\includegraphics[width=0.4\textwidth]{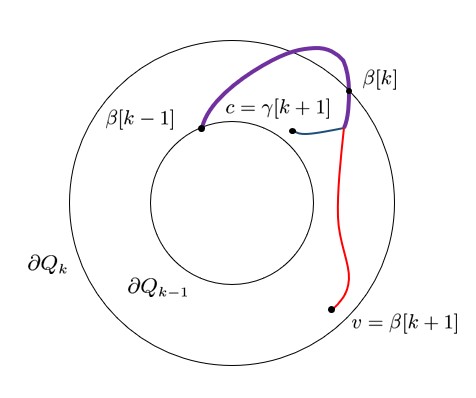}}
}
\subfloat[Case 1.b.]{
{\includegraphics[width=0.4\textwidth]{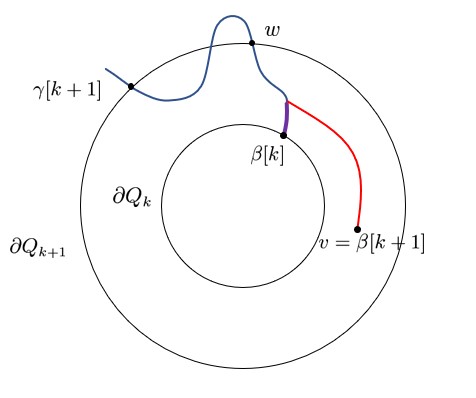}}
}
\newline
\subfloat[Case 1.c.]{
{\includegraphics[width=0.44\textwidth]{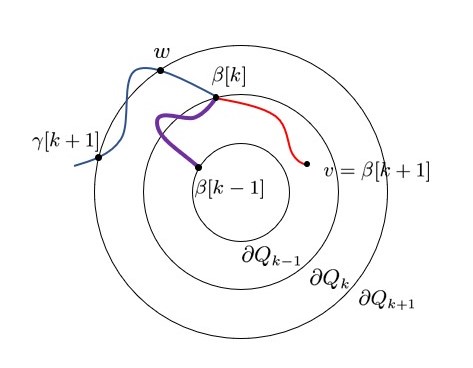}}
}
\subfloat[Case 2.a.]{
{\includegraphics[width=0.4\textwidth]{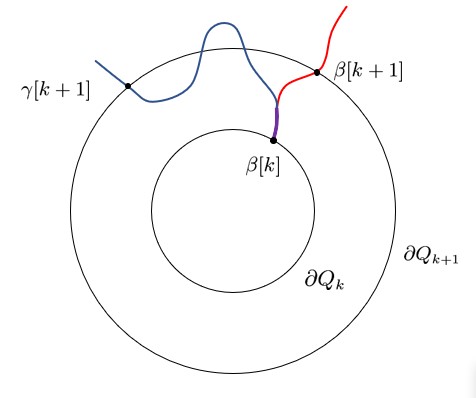}}
}
\newline
\subfloat[Case 2.b.i.]{
{\includegraphics[width=0.44\textwidth]{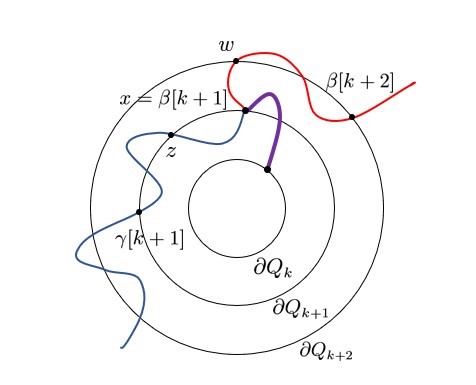}}
}
\subfloat[Case 2.b.ii]{
{\includegraphics[width=0.4\textwidth]{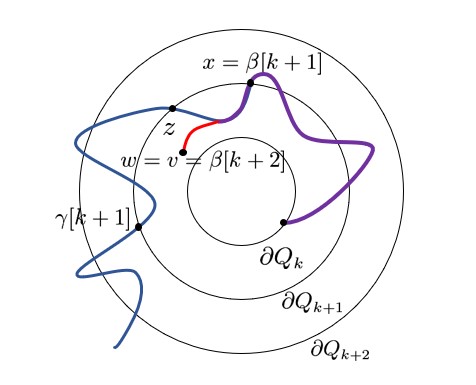}}
}
\caption{An illustration of the different cases that arise in determining the left/right/ancestor relationship of $\beta[k+1]$ and $\gamma[k+1]$ in the shortest path tree rooted at $s$ in $\out{R}$. The $s$-to-$v$ path is red, the $s$-to-$c$ path is blue, and their common prefix is purple.}\label{fig:cases}
\end{figure}

\textbf{Query Correctness.} 
Let us consider an invocation of \textsc{Dist} at level $i$. We first show that if such a query is answered without invoking \textsc{ModifiedPointLocate} (\cref{psdist2:modpointloc} of~\cref{alg:dist2}) then the distance it returns is correct. Indeed, let us recall that the input to \textsc{Dist} is a piece $Q_i$, a vertex $u \in \partial Q_i$ and a vertex $v \in \out{Q_i}$ and the desired output is the $u$-to-$v$ distance in $\out{Q_i}$ (along with some vertices on the $u$-to-$v$ shortest path in $\out{Q_i}$). 
Let us now look at the cases in which \textsc{Dist} does not call \textsc{ModifiedPointLocate} (Lines \ref{psdist2:start1}-\ref{psdist2:end1}). 
If $v \in Q_{i+1}$ then we have two cases: either the shortest $u$-to-$v$ path in $\out{Q_i}$ crosses $\partial Q_{i+1}$ or it does not. In the first case the answer is obtained by a point location query in the graph $Q_{i+1} \setminus (Q_i \setminus \partial Q_i)$ with sites the vertices of $\partial Q_{i+1}$ and additive weights the distances from $u$ to these vertices in $\out{Q_i}$, while in the second case by a query on the MSSP data structure for this graph with sources the vertices of $\partial Q_i$.
Else, if $v \not\in Q_{i+1}$ and $i=m-2$, the shortest path must cross $\partial Q_{i+1}$. We can thus obtain the distance by performing a point location query in the Voronoi diagram for $\out{Q_{m-1}}$ with sites the vertices of $\partial Q_{i+1}$ and additive weights the distances to these vertices from $u$ in $\out{Q_{m-2}}$.
Note that in both the cases that we perform point location queries, we have MSSP information for the underlying graph on which the Voronoi diagram is defined and hence the correctness of the answer follows from the correctness of the point location mechanism designed in~\cite{vorexact}.

We now observe that if invocations of \textsc{Dist} at level $i+1$ return the correct answer, then so 
does the invocation of \textsc{ModifiedPointLocate} at line~\ref{psdist2:modpointloc} of the invocation of \textsc{Dist} at level $i$,
since the only other requirement for the correctness of \textsc{ModifiedPointLocate} is is the ability to resolve the left/right/ancestor relationship; this follows from inspecting the cases analyzed above.
It follows that invocations of \textsc{Dist} at level $i$ also return the correct answer.
The correctness of the query algorithm then follows by induction. 

\textbf{Query time.}
We initially call \textsc{Dist} at level $0$. 
A call to \textsc{Dist} at level $i$, either returns the answer in $\cO(\log^2 n)$ time by using MSSP information and/or invoking \textsc{PointLocate} (Lines \ref{psdist2:start1}-\ref{psdist2:end1}), or makes a call to \textsc{ModifiedPointLocate} at level $i+1$ (\cref{psdist2:modpointloc}). For the latter to happen it must be that $i \leq m-3$.
This call to \textsc{ModifiedPointLocate} at level $i+1$ invokes a call to \textsc{ModifiedHandleCentroid} at level $i+1$.
Then, \textsc{ModifiedHandleCentroid} makes $\cO(\log r_{i+1})$ calls to \textsc{Dist} at level $i+1$, while going down the centroid decomposition of a Voronoi diagram with $\cO(\sqrt{r_{i+1}})$ sites.

Let $D[i]$ be the time required by a call to \textsc{Dist} at level $i$ and $C[i,t]$ be the time required by a call to \textsc{ModifiedHandleCentroid} at level $i$, where the size of the subtree of $\VD^*$ corresponding to the input centroid is $t$. We then have the following:
$$D[i]= \begin{cases}
	C[i+1,r_{i+1}]+\cO(\log n+\log^2 {r_{i+1}}) &  \text{If $i < m-2$} \\
	\cO(\log^2 n) & \text{If $i=m-2$}
\end{cases}
$$
(The additive $\cO(\log n)$ factor comes from the possible existence of multiple holes, see~\cref{sec:removingholes}.)

$$C[i,k]=
\begin{cases}
\cO(D[i]+\log n) + C[i,k/2] & \text{If $k>1$} \\
2D[i] & \text{If $k=1$}	
\end{cases}
$$
(The additive $\cO(\log n)$ factor comes from resolving the left/right/ancestor relationship.)

\noindent Hence, for $i<m-2$, $D[i]=\cO(\log r_{i+1}(D[i+1]+\log n))$.
The total time complexity of the query is thus $D[0]=\cO(\log^2 n \prod\limits_{i=1}^{m-3} c\log r_{i+1})=\cO(c^m \log^{m-1} n)$, where $c>0$ is a constant.

\ifspreport
\textbf{Retrieving the shortest path.} 
First of all, let us recall that given an MSSP data structure for a planar graph $H$ of size $n$ with sources the nodes of a face $f$, we can construct the shortest path from a node $x$ in $f$ to any node $y$ of $H$, backwards from $y$ in $\cO(\log\log n)$ per reported edge (cf.~\cite{PlanarBook}).
We can easily retrieve one by one the edges of the coarse trees $T^{Q_i}_{\beta[i]}$ corresponding to paths whose concatenation would give us the shortest $u$-to-$v$ path. The shortest path is in fact the concatenation of the $\beta[i]$-to-$\beta[i+1]$ paths in $T^{Q_i}_{\beta[i]}$, plus perhaps an extra final leg, which lies entirely in $Q_j \setminus (Q_{j-1} \setminus \partial Q_{j-1})$ for some $j$, if $v$ was appended to $\beta$.
Our MSSP data structures for each $i$ allow us to refine those coarse edges of $T^{Q_i}_{\beta[i]}$ when their underlying shortest path  lies entirely within $Q_{i+1} \setminus (Q_{i} \setminus \partial Q_{i})$.
For an edge $(x,y)$ of $T^{Q_i}_{\beta[i]}$ that corresponds to a path lying entirely in $\out{Q_{i+1}}$, we use the $x$-to-$y$ path in $T^{Q_{i+1}}_x$, which we can refine using the MSSP information we have for $Q_{i+2} \setminus (Q_{i+1} \setminus \partial Q_{i+1})$. We then proceed inductively on coarse edges whose underlying path is entirely in $\out{Q_{i+2}}$. In each step we obtain and refine an edge from a node in $\partial Q_i$ to a node in $\partial Q_{j}$, $j \in \{i-1,i,i+1\}$ using  the local MSSP information.
To summarize, we retrieve one by one the edges of the coarse trees $T^{Q_i}_{\beta[i]}$, corresponding to portions of the actual shortest path that lie entirely in $Q_i \setminus (Q_{i-1} \setminus \partial Q_{i-1})$ in $\cO(1)$ time each.
We then use the local MSSP information to obtain the actual edges of each coarse edge in $\cO(\log\log n)$ time per actual edge.

\bigskip
\fi

We prove in~\cref{sec:constr} that the time bound for the construction of the data structure is $\cOtilde(n \sum_i\frac{r_{i+1}}{\sqrt r_i})$. We summarize our main result in the following theorem.

\begin{restatable}{theorem}{mainthm}\label{main}
Given a planar graph $G$ of size $n$, for any decreasing sequence $n=r_m>r_{m-1}> \cdots > r_1$, where $r_1$ is a sufficiently large constant and $m \geq 3$, there is a data structure of size $\cO(n \sum_i\frac{r_{i+1}}{r_i} \log r_{i+1})$ that answers distance queries in time $\cO(c^m \log^{m-1} n)$, where $c>0$ is a constant. This data structure can be constructed in time $\cOtilde(n \sum_i\frac{r_{i+1}}{\sqrt r_i})$. 
\ifspreport
The shortest path can be retrieved backwards from $v$ to $u$, at a cost of $\cO(\log\log n)$ time per edge.
\fi
\end{restatable}

\noindent We next present three specific tradeoffs that follow from this theorem.

\begin{restatable}{corollary}{maincor}\label{cor:tradeoffs}
Given a planar graph $G$ of size $n$, we can construct a distance oracle admitting any of the following $\langle$space, query-time$\rangle$ tradeoffs:
\begin{enumerate}[label=({\alph*})]
\item $\langle \cO(n \log^{2+1/\epsilon} n),\cOtilde(n^{\epsilon})\rangle$, for any constant $\epsilon>0$;
\item $\langle n^{1+o(1)},n^{o(1)}\rangle$;
\item $\langle \cOtilde(n^{1+\epsilon}),\cO(\log^{1/\epsilon} n)\rangle$, for any constant $1/2 \geq \epsilon>0$.
\end{enumerate}
The construction time is $\cOtilde(n^{3/2} \log^{1/2\epsilon} n)$, $n^{3/2+o(1)}$ and $\cOtilde(n^{3/2+\epsilon})$, respectively.
\end{restatable}
\begin{proof}
Let $\rho_i$ denote $\frac{r_{i+1}}{r_i}$; we will set all $\rho_i$ to be equal to a value $\rho$.
We get tradeoffs (a), (b) and (c) by setting $\rho=\log^{1/\epsilon} n$,  $\rho=2^{\sqrt{\log n}}$ and $\rho=n^\epsilon$, respectively.

We provide the arithmetic below.

\begin{enumerate}[label=({\alph*})]
\item If we set $\rho=\log^{1/\epsilon} n$, then the depth of the recursive $r$-division is $m=\log_{\rho}n+1=\frac{\log n}{\log \rho}+1=\epsilon \log n/\log\log n+1$. Then, the space required is $\cO(n m \rho \log n )=\cO(n \log^{2+1/\epsilon} n)$ and the time required to answer a distance query is $\cOtilde((c\log n)^{\log_{\rho}n})=\cOtilde(n^{\epsilon})$.
\item If we set $\rho=2^{\sqrt{\log n}}$, then the depth of the recursive $r$-division is $m=\log_{\rho}n+1=\frac{\log n}{\log \rho}+1=\sqrt{\log n}+1$. Then, the space required is $\cO(nm \rho \log n)=n^{1+o(1)}$ and the time required to answer a distance query is $\cOtilde(\log^{\log_{\rho}n}n)=\cOtilde(\log^{\sqrt{\log n}}n)=\cOtilde(2^{\sqrt{\log n}\log\log n})=n^{o(1)}$.
\item If we set $\rho=n^\epsilon$, then the depth of the recursive $r$-division is $m=\log_{\rho}n+1=1/\epsilon+1$. Then, the space required is $\cO(n m\rho \log n)=\cO(n^{1+\epsilon}\log n)$ and the time required to answer a distance query is $\cO(\log^{\log_{\rho}n}n)=\cO(\log^{1/\epsilon}n)$.
\end{enumerate}
\end{proof}

\subsection{Dealing with multiple holes}
\label{sec:removingholes}
We now remove the assumption that all boundary vertices of each piece lie on a single hole. 

\paragraph{Data Structure.} The oracle consists of the following for each $0 \leq i \leq m-1$, for each piece $R \in \mathcal{R}_i$ whose parent in $\TG$ is $Q \in \mathcal{R}_{i+1}$:
\begin{enumerate}
\item If $i>0$, for each pair of holes $h$ of $R$ and $g$ of $Q$, such that $g$ lies in $R^{h,out}$, two MSSP data structures for $R^{h,out} \cap Q$, one with sources the vertices of $\partial R$ that lie on $h$, and one with sources the vertices of $\partial Q$ that lie on $g$.
\item If $i<m-1$, for each boundary vertex $u$ of $R$ that lies on a hole $h$ of $R$:
\begin{itemize} 
\item For each hole $g$ of $Q$ that lies in $R^{h,out}$:
\begin{itemize}
\item $\VDin^*(u,Q,g)$, the dual representation of the Voronoi diagram for $R^{h,out} \cap Q$ with sites the nodes of $\partial Q$ that lie on $g$, and additive weights the distances from $u$ to these nodes in $R^{h,out}$;
\item $\VDout^*(u,Q,g)$, the dual representation of the Voronoi diagram for $Q^{g,out}$ with sites the nodes of $\partial Q$ that lie on $g$, and additive weights the distances from $u$ to these nodes in $R^{h,out}$;
\end{itemize}

\item If $i>0$, the coarse tree $T_u^{R,h}$, which is the tree obtained from the shortest path tree rooted at $u$ in $R^{h,out}$ by contracting any edge incident to a vertex that is neither in $h$ nor in $\partial Q$, preprocessed as in the single hole case. 
\end{itemize}
\end{enumerate}

\paragraph{Query.} If $v \in Q_{i+1} \setminus Q_i$ then in Line~\ref{psdist2:near} of \textsc{Dist}, we need to consider $\cO(1)$ Voronoi diagrams $\VDin^*(u,Q_{i+1},g)$, one for each hole of $Q_{i+1}$,  instead of just one. 
If not, we find the correct hole $h$ of $Q_{i+1}$ such that $v \in Q_{i+1}^{h,out}$, and recurse on $\VDout^*(u,Q_{i+1},h)$. 

We can find the correct hole by storing some information for each hole of each piece.
It is not hard to see that each separator of the $\cO(\log n)$ ancestors of a piece $P \in \AG$ lies in $P^{h,out}$ for some hole $h$ of $P$.
For each piece $P \in \AG$, we store, for each separator of an ancestor of $P$ in $\AG$, the hole $h$ of $P$ such that $P^{h,out}$ contains that separator. 
Then, by performing an LCA query for the constant size piece of $\AG$ containing $v$ and $Q_{i+1}$ in $\AG$ we find the separator that separated $v$ from $Q_{i+1}$ and we can thus find the appropriate hole of $Q_{i+1}$ in $\cO(1)$ 
time.

Obtaining the left/right/ancestor information is essentially the same as in the single hole case. There is one additional case that needs to be considered, in which the vertex at which the $s$-to-$v$ path and the $s$-to-$v^c$ path diverge is located inside a hole. 
In this case we may have to look into finer and finer coarse trees until we get to a level of resolution where we can either deduce the left/right relationship from the coarse tree at that level, or from the MSSP data structure at that level. 
This refining process has $\cO(\log n)$ levels and takes constant time per level, except for the last level where we use the MSSP data structure in $\cO(\log n)$ time as in the single hole case. Therefore, the asymptotic query time is unchanged.   

\ifspreport
As for retrieving the edges of the shortest path in $\cO(\log\log n)$ time each, the presence of multiple holes does not pose any further complications.
\fi

\paragraph{Dealing with holes that are non-simple cycles} 
Non-simple holes do not introduce any significant difficulties. The technique to handle non-simple holes is the one described in~\cite[Section~5.1, pp. 27]{DBLP:journals/talg/KaplanMNS17}. 
We make an incision along the non-simple boundary of a hole to turn it into a simple cycle. 
Making this incision creates multiple copies of vertices that are visited multiple times by the non-simple cycle. However, the total number of copies is within a constant of the original number of boundary vertices along the hole. 
These copies do not pose any problems because, for our purposes, they can be treated as distinct sites of the Voronoi diagram.

\section{Construction Time}\label{sec:constr}

In this section we show how to construct our oracles in  the $\cOtilde(n \sum_i\frac{r_{i+1}}{\sqrt r_i})$ time stated in Theorem~\ref{main}. Before doing so, we give some preliminaries on dense distance graphs. 

\paragraph{Dense distance graphs and FR-Dijkstra.}
The \emph{dense distance graph} of a piece $P$, denoted 
$DDG_P$ is a complete directed graph on the boundary vertices of $P$.
Each edge $(u,v)$ has weight $d_{P}(u,v)$, equal to the length of the shortest $u$-to-$v$ path in $P$.
$DDG_P$ can be computed in time $\cO(|\partial P|^2 + |P| \log |P|)$ using the multiple source shortest paths (MSSP) algorithm~\cite{MSSP,DBLP:journals/siamcomp/CabelloCE13}.
Over all pieces of the recursive decomposition this takes time $\cO(n \log^2 n)$ in total and requires space $\cO(n \log n)$.
We refer to the aforementioned (standard) DDGs as {\em internal}; the {\em external} DDG of a piece $P$ is the complete directed graph on the vertices of $\partial P$, with the edge $(u,v)$ having weight equal to $d_{\out{P}}(u,v)$. 
There is an efficient implementation of Dijkstra's algorithm (nicknamed FR-Dijkstra~\cite{FR}) that runs on any  union of $DDG$s. The algorithm exploits the fact that, due to planarity, a DDG's adjacency matrix $M$ satisfies a Monge property (namely, for any $i\neq j$ we have that $M[i+1,j] + M[i,j+1]~\le~ M[i+1,j+1] + M[i,j]$). We next give a --convenient for our purposes -- interface for FR-Dijkstra which was essentially proved in~\cite{FR}, with some additional components and details from~\cite{DBLP:journals/talg/KaplanMNS17,DBLP:conf/esa/MozesW10}.

\begin{theorem}[\cite{FR,DBLP:journals/talg/KaplanMNS17,DBLP:conf/esa/MozesW10}]\label{thm:FR}
A set of $DDG$s with $\cO(M)$ vertices in total (with multiplicities), each having at most $m$ vertices, can be preprocessed in time and space $\cO(M \log m)$ in total.
After this preprocessing, Dijkstra's algorithm can be run on the union of any subset of these $DDG$s with $\cO(N)$ vertices in total (with multiplicities) in time $\cO(N\log^2 m)$.
\end{theorem}

\paragraph{The construction.}
For ease of presentation we present the construction algorithm assuming pieces have only one hole. Generalizing to multiple holes does not pose any new obstacles. Recall that for each piece $R$ in the recursive decomposition whose parent in $\TG$ is $Q$ we need to compute two MSSP data structures, and for each boundary vertex $u$ of $R$ we need to compute $\VDin^*(u,Q)$, $\VDout^*(u,Q)$, and the coarse tree $T_u^R$.

Computing the MSSP data structures takes nearly linear time in the total size of these data structures.
To be able to compute the distances between vertices in the same piece of the $r_1$ division, additive distances for all Voronoi diagrams, and the coarse trees $T_u^R$, we compute the internal and external dense distance graphs of every piece in the recursive $r$-division. 
This can be done in nearly linear time~\cite{DBLP:conf/soda/MozesS12}. 
We also compute the dense distance graph of $\partial R \cup \partial Q$. 
This is done by running MSSP twice in $Q \setminus (R \setminus \partial R)$, and then querying all pairwise distances.
For a piece $R$ of the $r_i$ division, this takes $\tilde \cO(r_{i+1})$ time, so a total of $\tilde \cO(\sum_i n\frac{r_{i+1}}{r_i})$ time overall. 

We then compute the distances from each vertex in $G$ to all other vertices in the same piece $R$ of the $r_1$-division in time $\cO(1)$, using FR-Dijkstra on $R$ and the external DDG of $R$.
To compute the additive weights for the Voronoi diagrams $\VDin^*(u,Q)$ and $\VDout^*(u,Q)$, as well as the coarse tree $T_u^R$, we need to compute distances from $u$ to $\partial Q \cup \partial R$ in $\out{R}$. 
This can be done in time $\cO(|\partial Q|\log^2 |\partial Q|)$ by running FR-Dijkstra from $u$ on the union of the DDG of $Q \setminus (R \setminus \partial R)$ and the external DDG of $Q$. 
The total time to compute all such additive weights is thus $\cO(\sum_i n \sqrt{\frac{r_{i+1}}{r_i}}\log^2 r_{i+1} )$.

To compute each Voronoi diagram $\VDin^*(u,Q)$ we compute the primal Voronoi diagram $\VDin(u,Q)$ with a single-source shortest path computation in $Q$ from an artificial super-source $s$ connected to all vertices $\partial Q$ with edges whose weights are the additive weights. This can be done in linear $O(|Q|)$ time~\cite{DBLP:journals/jcss/HenzingerKRS97}. In $O(|Q|)$ time it is then easy to obtain the dual $\VDin^*(u,Q)$ from the primal $\VDin(u,Q)$. The total time to compute all such diagrams is therefore $\cO(\sum_i n \frac{r_{i+1}}{\sqrt{r_i}})$.

It remains to compute the diagram $\VDout^*(u,Q)$. Here we cannot afford to make an explicit computation of the primal Voronoi diagram $\VDout(u,Q)$. Instead, we will compute just the tree structure and the trichromatic faces of $\VDout(u,Q)$. 
These suffice as the representation of $\VDout^*(u,Q)$ since all we need for point location is the centroid decomposition of $\VDout^*(u,Q)$. The main idea is to use FR-Dijkstra to locate the trichromatic faces. 
Let $K$ be a star with center $u$ and leaves $\partial Q$, where for each vertex $w \in \partial Q$,  the length of the arc $uw$ is the additive distance of $w$ in $\VDout^*(u,Q)$.
We start by computing a shortest path tree $T$ rooted at $u$ in the union of $K$ and the DDGs of all siblings of pieces in the complete recursive decomposition tree $\AG$ that contain $Q$. We shall prove that, by inspecting the restriction of $T$ to the DDG of each piece $P$, we can infer whether $P$ contains a trichromatic face or not. 
If $P$ contains a trichromatic face we refine $T$ inside $P$ by replacing the DDG of $P$ with the DDGs of its two children in $\AG$, and recomputing the part of $T$ inside $P$ using FR-Dijkstra. 
We continue doing so until we locate all $O(|\partial Q|)$ trichromatic faces. The tree structure of $\VDout(u,Q)$ is captured by the structure of the shortest path tree in the DDGs of all the pieces at the end of this process.
The total time to locate all the trichromatic faces is proportional, up to polylogarithmic factors, to the total number of vertices in all of the DDGs involved in all these computations.
Since the total number of vertices in all DDGs containing a particular face is $\cO(\sqrt n)$, the total time for finding all trichromatic faces as well as the tree structure is $\cOtilde(\sqrt n \cdot \sqrt{|Q|})$. 
A more careful analysis that takes into account double counting of large pieces (cf.~\cite{faultyOracle}) bounds this time by $\cOtilde(\sqrt{n \sqrt {|Q|}})$. Since there are $n/\sqrt r_i$ boundary vertices at level $i$ of the recursive $r$-division, the total time to compute all Voronoi diagrams $\VDout^*(u,Q)$ for a single level is $\cOtilde(\frac{n}{\sqrt r_i} \cdot \sqrt{n} r_i^{1/4}) = \cO(n^{3/2})$. This is dominated by the  $\cO(\sum_i n \frac{r_{i+1}}{\sqrt{r_i}})$ term for any choice of $r_i$'s.

The only missing part in the explanation above is showing we can infer whether a piece $P$ contains a trichromatic face just by inspecting the restriction of $T$ to $P$. 
Our choice of additive distance guarantees that each vertex of $\partial Q$ is a child of $u$ in $T$. We label each vertex $w$ of $T$ by its unique ancestor in $T$ that belongs to $\partial Q$. Note that the label of a vertex $w$ corresponds to the Voronoi cell containing $w$ in $\VDout^*(u,Q)$.
Consider the restriction of $T$ to the DDG of $P$. 
We use a representation of size $\cO(|\partial P|)$ of the edges of $T$ embedded as curves in $P$, such that each edge of $T$ is homologous (w.r.t. the holes of $P$) to its underlying shortest path in $P$. See~\cite{DBLP:conf/esa/LackiS11,DBLP:conf/soda/MozesNNW18} for details on such a representation. 
We make incisions in the embedding of $P$ along the edges of $T$ (the endpoints of edges of $T$ are duplicated in this process). Let $\mathcal P$ be the set of connected components of $P$ after all incisions are made.

 \begin{lemma}\label{lem:trichrom}
$P$ contains a trichromatic face if and only if some connected component $C$ in $\mathcal P$ contains boundary vertices of $P$ with more than two distinct labels.
\end{lemma}
Intuitively, for each connected component $C $ in $\mathcal P$, each label appears as the label of boundary vertices along at most a single sequence of consecutive boundary vertices along the boundary of $C$. Consider the component $C'$ obtained from $C$ by connecting an artificial new vertex to all the boundary vertices with the same label. Since these boundary vertices form a single consecutive interval on the boundary of $C$, $C'$ is also a planar graph when the new vertices are embedded in the infinite face of $C$. Now $C'$ has a new infinite face where every vertex on that face has a distinct label, and there are more than two labels. The Voronoi diagram of $C'$ necessarily contains a trichromatic face.

\begin{figure}[htpb]
\begin{center}
\includegraphics[width=0.65\textwidth]{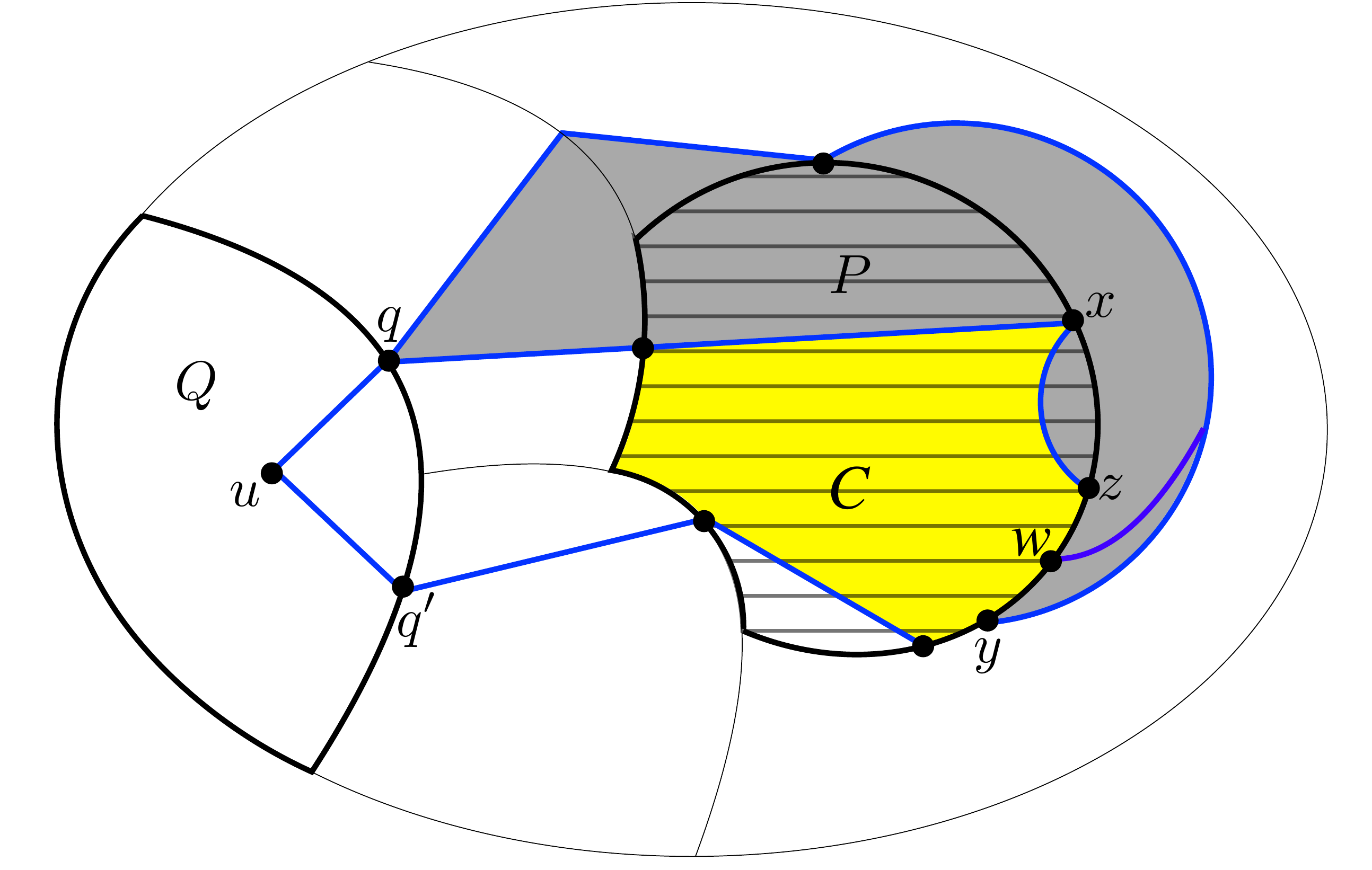}
\end{center}
\caption{Illustration for the proof of \cref{lem:trichrom}. Some pieces in a graph $G$ are shown. The piece $Q$ (in bold) is a piece of some $r_i$-division in the recursive $r$-division of $G$. The piece $P$ (bold boundary, horizontal stripes) is some piece in the complete recursive decomposition of $G$ that lies outside $Q$. The shortest path tree $T$ is shown in blue, the connected component $C$ in yellow, and the cycle $D$ in gray. Vertices $x$ and $y$ have the same label $q$. The vertices $w,z$ between $x$ and $y$ (on the cyclic walk $F$ along the infinite face of $C$) must also be labeled $q$.\label{fig:lemma8}}
\end{figure}

\begin{proof}
Let $C$ be a connected component in $\mathcal P$. 
Note that the vertices of $\partial Q$ either do not belong to $C$ or they are incident to a single face $f$ of $C$. In the former case, let $f$ be the face of $C$ such that $\partial Q$ is embedded in $f$. 
We think of $f$ as the infinite face of $C$. Note that, because any path from $\partial Q$ to any vertex of $C$ must intersect $f$, the set of  labels of the vertices of $f$ is identical to the set of labels of all of $C$. 

We first claim that the vertices of $f$ that have the same label are consecutive in the cyclic order of $f$. To see this, consider any two distinct vertices $x,y$ of $f$ that have the same label $q$. Note that, by the incision process defining $C$, if $x$ is an ancestor of $y$ in $T$ (or vice versa) then all the vertices between $x$ and $y$ on $f$ are on the $x$-to-$y$ path in $T$, and hence all have the label $q$. 
Assume, therefore, that neither is an ancestor of the other (as illustrated in Figure~\ref{fig:lemma8}), and consider the (not necessarily simple) cycle $D$ (in $\out Q$) formed by the unique $x$-to-$y$ path in $T$, and the $x$-to-$y$ path $F$ of $f$ such that $D$ does not enclose $C$. Observe that $D$ is non self crossing, and that $D$ encloses no vertices of $\partial Q$ except, perhaps, $q$ itself.
Suppose some vertex of $F$ has a label $q' \neq q$, and let $w$ be a rootmost such vertex in $T$. Since $D$ does not enclose $q'$, the $q'$-to-$w$ path in $T$ is not enclosed by $D$, and hence must intersect $C$. But then $C$ should have been further dissected when the incisions along $T$ were performed, a contradiction. 

The argument above established that the vertices of $f$ that have the same label are consecutive in the cyclic order of $f$.
For every unique label $q$ of a vertex of $f$ we embed 
 inside $f$ an artificial vertex and connect it to every vertex $w$ of $f$ that has label $q$ with an edge whose length is the $q$-to-$w$ distance in $T$. By the consecutiveness property above, this can be done without violating planarity. 
 We connect the artificial vertices by infinite length edges, according to the cyclic order of the labels along $f$, so that resulting graph $C'$ has a new infinite face containing just the artificial vertices. 
 Consider now the additively weighted Voronoi diagram $\VD(C')$ with the artificial vertices as sites, where the additive weight of an artificial vertex with label $q$ is equal to the additive weight of the corresponding site $q$ in $\VDout^*(u,Q)$. Note that, by construction, the additive distances in $\VD(C')$ and in $\VDout^*(u,Q)$ are the same, so the restrictions of both diagrams to $C$ are identical. It therefore suffices to prove that $\VD(C')$ has a trichromatic face in $C$.
 
Since $C$ has boundary vertices of $P$ with more than two distinct labels, $C'$ has more than two sites. Since all the vertices on the infinite face of $C'$ are sites, and since every site is in its own Voronoi cell, $\VD(C')$ has at least two trichromatic faces~\cite{DBLP:conf/soda/Cabello17,vorexact} (one being the infinite face of $C'$). Observe that all the newly introduced faces of $C'$ (those created by the edges connecting artificial vertices to vertices of $f$ or to each other) have either a single label (in case the face is a triangle formed by an artificial vertex and two consecutive vertices of $f$ with the same label), or two labels (in case the face has size 4, and is formed by two artificial vertices and two consecutive vertices of $f$ with two distinct labels). It follows that $\VD(C')$ must have a trichromatic face of $C$, and so does $\VDout^*(u,Q)$.     
\end{proof}

\section{Final remarks}
The main open question is whether there exists a distance oracle occupying linear (or nearly linear) space and requiring constant (or polylogarithmic) time to answer queries. Note that currently, any oracle with $\cOtilde(n)$ space requires polynomial query time. In particular, the fastest known oracles with strictly linear space~\cite{DBLP:conf/soda/MozesS12,DBLP:conf/wads/Nussbaum11} require $\Omega(n^{1/2+\epsilon})$ query time.
Another important question concerns the construction time. Is there a nearly-linear time algorithm to construct our oracle, or oracles with comparable space to query-time tradeoffs?

\bibliographystyle{plain}

\end{document}